\documentclass[final,leqno,onethmnum]{siamltex}
\usepackage{mathrsfs}

\usepackage{amsmath,amssymb,mathrsfs,amsfonts,graphicx,latexsym,exscale,cmmib57,dsfont,bbm,amscd,ulem}

\def\bbalpha{{\boldsymbol\alpha}}
\def\biota{{\boldsymbol\iota}}
\def\bbell{{\boldsymbol\ell}}
\def\bbchi{{\boldsymbol\chi}}
\def\bbsigma{{\boldsymbol\sigma}}
\def\bblambda{{\boldsymbol\lambda}}
\def\bbbeta{{\boldsymbol\beta}}

\def\bbdelta{{\boldsymbol\delta}}
\def\bomega{{\boldsymbol\omega}}
\def\bbA{{\boldsymbol A}}
\def\bbC{{\boldsymbol C}}
\def\bL{{\boldsymbol L}}
\def\bcdot{{\boldsymbol\cdot}}
\def\bb{{\boldsymbol b}}
\def\be{{\boldsymbol e}}
\def\bu{{\boldsymbol u}}
\usepackage[colorlinks=true]{hyperref}

\newtheorem{thm}{\textbf{Theorem}}[section]
\newtheorem{prop}[thm]{\textbf{Proposition}}
\newtheorem{cor}[thm]{\textbf{Corollary}}
\newtheorem{lem}[thm]{\textbf{Lemma}}
\newtheorem{que}[thm]{\textbf{Question}}

\newtheorem{defn*}{\textbf{Definitions}}
\newtheorem{defn}[thm]{\textbf{Definition}}
\newtheorem{remark}[thm]{\textbf{Remark}}

\title{Criteria of Stabilizability for Switching-control Systems with Solvable Linear Approximations\thanks{Received by the editors August 19, 2010; accepted
for publication (in revised form) January 9, 2012; published
electronically XX XX, 2012. This work was supported in part by National Natural Science Foundation of China
(Grant No. 11071112) and PAPD of Jiangsu Higher Education.
\URL sicon/xx-x/xxxx.html}
}%

\author{Xiongping Dai\thanks{Department of Mathematics, Nanjing University, Nanjing 210093, People's Republic of China
(xpdai@nju.edu.cn).}
}%

\begin{document}
\slugger{sicon}{2011}{XX}{X}{1--32}%

\maketitle

\setcounter{page}{1}

\begin{abstract}
We study the stability and stabilizability of a continuous-time switched control system that consists of the time-invariant $n$-dimensional subsystems
\begin{equation*}
\dot{x}=A_ix+B_i(x)u\quad (x\in\mathbb{R}^n, t\in\mathbb{R}_+\textrm{ and }u\in\mathbb{R}^{m_i}),\qquad
\textrm{where }i\in\{1,\dotsc,N\}
\end{equation*}
and a switching signal $\sigma(\bcdot)\colon\mathbb{R}_+\rightarrow\{1,\dotsc,N\}$ which orchestrates switching between these subsystems above, where $A_i\in\mathbb{R}^{n\times n}, n\ge1, N\ge2, m_i\ge1$, and where $B_i(\bcdot)\colon\mathbb{R}^n\rightarrow\mathbb{R}^{n\times m_i}$ satisfies the
condition $\|B_i(x)\|\le\bbbeta\|x\|\;\forall x\in\mathbb{R}^n$. We show that, if $\{A_1,\dotsc,A_N\}$ generates a solvable Lie algebra over the field $\mathbbm{C}$ of complex numbers and there exists an element $\bbA$ in the convex hull $\mathrm{co}\{A_1,\dotsc,A_N\}$ in $\mathbb{R}^{n\times n}$ such that the affine system $\dot{x}=\bbA x$ is exponentially stable, then there is a constant $\bbdelta>0$ for which one can design ``sufficiently many" piecewise-constant switching signals $\sigma(t)$ so that the switching-control systems
\begin{equation*}
\dot{x}(t)=A_{\sigma(t)}x(t)+B_{\sigma(t)}(x(t))u(t),\quad x(0)\in\mathbb{R}^n\textrm{ and }
t\in\mathbb{R}_+
\end{equation*}
are globally exponentially stable, for any measurable external inputs $u(t)\in\mathbb{R}^{m_{\sigma(t)}}$ with $\|u(t)\|\le\bbdelta$.
\end{abstract}
\begin{keywords}
Switching-control system, stabilizability, Lyapunov exponent, Liao-type exponent
\end{keywords}

\begin{AMS}
Primary: 93C15, 34H05 Secondary: 37N35, 93D20, 93D15, 93C73
\end{AMS}
\begin{DOI}
xx.xx/xxxxxx
\end{DOI}

\pagestyle{myheadings}

\thispagestyle{plain}

\markboth{\sc{X.~Dai}}{Criteria of Stability and Stabilizability}
\section{Introduction}\label{sec1}%
Let $\mathbb{R}^n$ be the real $n$-dimensional Euclidean space with an inner product $\langle\cdot,\cdot\rangle$ which gives rise to a vector norm $\|\cdot\|$ on it.
In this paper, we will focus on the stability and stabilizability issues for the continuous-time switched control system
\begin{align}
&\dot{x}(t)=A_{\sigma(t)}x(t)+B_{\sigma(t)}(x(t))u&& (x(0)\in\mathbb{R}^n\textrm{ and } t>0),\label{eq1.1}\\
\intertext{the subsystems of which are time-invariant continuous-time control systems}
&\dot{x}(t)=A_ix(t)+B_i(x(t))u&& (x(0)\in\mathbb{R}^n, t>0\textrm{ and }u\in\mathbb{R}^{m_i})\label{eq1.2}\\
\intertext{for $i\in\{1,\dotsc,N\}$, where $A_i\in\mathbb{R}^{n\times n}$ for all indices $i$, and $\sigma\colon(0,+\infty)\rightarrow\{1,\dotsc,N\}$
is piecewise constant and left-continuous having at most finite number of discontinuities on any finite interval of $\mathbb{R}_+:=(0,\infty)$. Here $n\ge1, N\ge2$ and $m_i\ge1$ all are integers. We assume that the matrix-valued functions $B_i(x)\in\mathbb{R}^{n\times m_i}$ are continuous with respect to $x\in\mathbb{R}^n$ satisfying the linear growth condition:}
&\|B_i(x)\|\le\bbbeta\|x\|&&\forall x\in\mathbb{R}^n\label{eq1.3}
\end{align}
for each $i\in\{1,\dotsc,N\}$, for some constant $\bbbeta>0$. However $B_i(x)$ does not need to be Lipschitz continuous, not even locally, with respect to $x\in\mathbb{R}^n$.

For any given $\sigma(t)$ and $u(t)$, the switched control system (\ref{eq1.1}) is said to be (globally) \textit{exponentially stable}, provided that for any given initial state $x_0\in\mathbb{R}^n$, its solutions $x(t)=\phi(t,x_0,\sigma,u)$ with
$x(0)=x_0$, that are absolutely continuous in $t$ but not necessarily unique because of the lack of the Lipschitz condition of the nonlinear terms $B_i(x)$, are such that
\begin{equation*}
\pmb{\chi}^+(x_0,\sigma,u):=\limsup_{t\to+\infty}\frac{1}{t}\log\|\phi(t,x_0,\sigma,u)\|<0
\end{equation*}
if $x(t)$ is forwardly complete, i.e., $x(t)$ may be extended on $\mathbb{R}_+$.

The stability issues of such switched systems include several interesting phenomena.
For example, even when all the subsystems (\ref{eq1.2}) are exponentially stable, (\ref{eq1.1}) may have divergent trajectories for certain switching signals $\sigma(t)$;
see, e.g. \cite{Branicky98, Margaliot06}. Another noticeable fact is that one may carefully switch between unstable subsystems to make (\ref{eq1.1}) exponentially stable; see, e.g. \cite{Utkin77, DBPLA}. As these examples suggest, the stability of switched systems depends not only upon the dynamics of each subsystems but also upon the properties of the switching signals. Therefore, the stability study of switched systems might be roughly divided into two kinds of problems~\cite{LA09}:
\begin{enumerate}
\item[$(\mathds{Q}1)$] one is the stability analysis of switched systems under given sets of admissible switching signals (all switching signals or switching signals obeying some constraints);

\item[$(\mathds{Q}2)$] the other is the synthesis of stabilizing switching signals for a given collection of dynamical/control systems.
\end{enumerate}

\noindent In the present paper, the question that we are concerned with is a complex of the above two kinds of problems. In our context, all subsystems (\ref{eq1.2}) are not necessarily exponentially stable themselves, but there exists an exponentially stable system in their convex hull $\mathrm{co}\{A_1,\dotsc,A_N\}$ in $\mathbb{R}^{n\times n}$. For this, we want to seek some kind of condition that may guarantee the existence of switching systems (\ref{eq1.1}) that are globally exponentially stable; and further to describe such stable switching signals.

To describe the switching signals that we are of interest to goal here, we need to introduce the classical symbolic space. Let
\begin{equation}\label{eq1.4}
\varSigma_{\!N}^+=\left\{\biota=(\biota_k)_{k=1}^{+\infty}\,|\,\biota_k\in\{1,\dotsc,N\}\, \forall k\ge1\right\}
\end{equation}
be the one-sided symbolic sequence space, which is compact and metrizable, endowed with the standard product topology. Then, there gives rise to the canonical symbolic dynamical system\,---\,the one-sided Markovian (forward) shift transformation:
\begin{equation}\label{eq1.5}
\theta\colon\varSigma_{\!N}^+\rightarrow\varSigma_{\!N}^+;\quad (\biota_k)_{k=1}^{+\infty}\mapsto(\biota_{k+1})_{k=1}^{+\infty}.
\end{equation}
Observe that the shift $\theta$ is continuous and surjective, not generally $1$-to-$1$. For any vector $\vec{\alpha}=(\alpha_1,\dotsc,\alpha_N)\in\mathbb{R}^N$ with $0<\alpha_i<1$ and $\alpha_1+\cdots+\alpha_N=1$, we can naturally define a probability measure/distribution, written as $\mathds{P}_{\vec{\alpha}}$, on $\varSigma_{\!N}^+$ in this way: for any cylinder sets of length $k\ge1$
\begin{equation}\label{eq1.6}
[i_1,\dotsc,i_{k}]:=\left\{\biota\in\varSigma_{\!N}^+\,|\,\biota_1=i_1,\dotsc,\biota_{k}=i_{k}\right\},
\end{equation}
we have
\begin{equation}\label{eq1.7}
\mathds{P}_{\vec{\alpha}}([i_1,\dotsc,i_{k}])=\alpha_{i_1}\cdots\alpha_{i_{k}},
\end{equation}
for every words $(i_1,\dotsc,i_{k})\in\{1,\dotsc,N\}^k$. Then, $(\varSigma_{\!N}^+,\theta,\mathds{P}_{\vec{\alpha}})$ is an ergodic dynamical system, see e.g. \cite{NS,Walters82}; that is to say, firstly $\mathds{P}_{\vec{\alpha}}(\theta^{-1}B)=\mathds{P}_{\vec{\alpha}}(B)$ for any Borel subsets $B$ of $\varSigma_{\!N}^+$, and secondly $\mathds{P}_{\vec{\alpha}}(B)=0$ or $1$ if the probability of the symmetric difference $\mathds{P}_{\vec{\alpha}}(B\triangle\theta^{-1}B)=0$. In addition, it is easy to see that every $\biota\in\varSigma_{\!N}^+$ is a density point of $\mathds{P}_{\vec{\alpha}}$; that is, every neighborhood of $\biota$ has positive measure of $\mathds{P}_{\vec{\alpha}}$.

Now, we make a convention for our convenience: To any $\biota\in\varSigma_{\!N}^+$, there corresponds a continuous-time piecewise constant left-continuous switching signal
\begin{equation}\label{eq1.8}
\sigma_{\biota}(\bcdot)\colon\mathbb{R}_+\rightarrow\{1,\dotsc,N\};\quad \sigma_{\biota}(t)=
\biota_k\;\textrm{whenever }k-1<t\le k\;\forall k\in\mathbb{N},
\end{equation}
where and in the sequel $\mathbb{N}=\{1,2,\dotsc\}$. Here we do not care the value $\sigma_{\biota}(0)$, since the initial value $x(0)=x_0$ is given previously.

Then, associated to the collection of subsystems given as in (\ref{eq1.2}), there generates a switching-control dynamical system
\begin{equation*}
\dot{x}=A_{\sigma_{\biota}(t)}x+B_{\sigma_{\biota}(t)}(x)u,\quad x(0)\in\mathbb{R}^n, t\in\mathbb{R}_+,\textrm{ and } u\in\mathbb{R}^{m_{\sigma_{\biota}(t)}},\leqno{(\ref{eq1.1})_{\biota}}
\end{equation*}
for any $\biota\in\varSigma_{\!N}^+$. Our problem concerned here can now be stated as follows:

\begin{que}\label{que1.1}
Let $\vec{\alpha}=(\alpha_1,\dotsc,\alpha_N)\in\mathbb{R}^N$ be a positive probability vector; that is, $0<\alpha_k<1$ for $1\le k\le N$ and $\alpha_1+\cdots+\alpha_N=1$.
\begin{enumerate}
\item[$(\mathbbm{1})$] If the linear affine equation
\begin{equation*}
\dot{x}(t)=(\alpha_1A_1+\cdots+\alpha_NA_N)x(t),\quad x(0)\in\mathbb{R}^n\textrm{ and } t\in\mathbb{R}_+
\end{equation*}
is exponentially stable then, can one design switching signals $\sigma_{\biota}(t)$ with low switching frequency so that the corresponding systems $(\ref{eq1.1})_{\biota}$ steered by $\sigma_{\biota}(t)$,
are exponentially stable?

\item[$(\mathbbm{2})$] What condition can guarantee that $(\ref{eq1.1})_{\biota}$ are globally exponentially stable for $\mathds{P}_{\vec{\alpha}}$-almost sure $\biota\in\varSigma_{\!N}^+$?
\end{enumerate}
\end{que}

\noindent Question~\ref{que1.1}.($\mathbbm{1}$), corresponding to the above $(\mathds{Q}2)$, is to find switching signals $\sigma_\biota(t)$ to steer the switched systems $(\ref{eq1.1})_\biota$ globally exponentially stable. And Question~\ref{que1.1}.($\mathbbm{2}$), corresponding to the above $(\mathds{Q}1)$, is also one of the fundamental problems for the stability analysis of switched systems.
Here we will give a unified positive solution to this question under an additional algebraic condition\,---\,solvability.

Let $\mathcal{A}=\{A_1,\dotsc,A_N\}\subset\mathbb{R}^{n\times n}$ be arbitrarily given. Then under the Lie bracket $[A_i,A_j]=A_iA_j-A_jA_i$ for all $A_i,A_j\in\mathbb{R}^{n\times n}$, $\mathcal{A}$ generates a Lie algebra, write $\mathcal{A}_{\mathrm{LA}}$, over the field $\mathbbm{C}$ of complex numbers; that is the smallest Lie algebra containing $\mathcal{A}$ over the field $\mathbbm{C}$. Letting
\begin{equation*}
\mathcal{A}_{\mathrm{LA}}^{(0)}=\mathcal{A}_{\mathrm{LA}},\ \mathcal{A}_{\mathrm{LA}}^{(1)}=\left[\mathcal{A}_{\mathrm{LA}}^{(0)},\mathcal{A}_{\mathrm{LA}}^{(0)}\right],\ \ldots,\
\mathcal{A}_{\mathrm{LA}}^{(\ell)}=\left[\mathcal{A}_{\mathrm{LA}}^{(\ell-1)},\mathcal{A}_{\mathrm{LA}}^{(\ell-1)}\right], \dotsc,
\end{equation*}
$\mathcal{A}$ is called {\it solvable} over $\mathbbm{C}$, provided that $\mathcal{A}_{\mathrm{LA}}^{(\ell)}=\{0_{n\times n}\}$ for some integer $\ell\ge1$, where $0_{n\times n}$ denotes the zero matrix in $\mathbbm{C}^{n\times n}$. For example, abelian or nilpotent $\mathcal{A}$ implies solvable; see, e.g., \cite{Hum}.

It is well known that the Lie algebra $\mathcal{A}_{\mathrm{LA}}$ plays a very important role in the theory of reachability and controllability, for example, see \cite{Sontag, BP}.
When each of the subsystems $A_i$ is exponentially stable, some stability criteria, for arbitrary piecewise constant switching signals, of (\ref{eq1.1}) have been developed under the solvability condition and an additional higher regularity of $B_i(x)$, such as smoothness or analyticity, with respect to the state-variable $x\in\mathbb{R}^n$; for example, see~\cite{NB94, Gurvits, SNS, LHM99, AL01, ML06, HBF}.

Under this algebraic solvability condition, our main result obtained in this paper can be formulated as follows:

\begin{thm}\label{thm1.2}
Assume that $\vec{\alpha}=(\alpha_1,\dotsc,\alpha_N)\in\mathbb{R}^N$ is a positive probability vector. Let $\mathcal{A}=\{A_1,\dotsc,A_N\}\subset\mathbb{R}^{n\times n}$ be solvable over the complex-number field $\mathbbm{C}$. If it holds that
\begin{equation*}
\dot{x}(t)=(\alpha_1A_1+\cdots+\alpha_NA_N)x(t),\; x(0)\in\mathbb{R}^n\textrm{ and } t\in\mathbb{R}_+,\textrm{ is exponentially stable}, \leqno{(\star)}
\end{equation*}
then there hold the following two statements.
\begin{description}
\item[(1)] For $\mathds{P}_{\vec{\alpha}}$-a.s.~$\biota\in\varSigma_{\!N}^+$, the linear switched systems
\begin{equation*}
\dot{x}(t)=A_{\sigma_{\biota}(t)}x(t),\quad x(0)\in\mathbb{R}^n\textrm{ and }t\in\mathbb{R}_+\leqno{(\ref{eq1.1})_{\biota,0}}
\end{equation*}
are exponentially stable.

\item[(2)] Moreover, if condition $(\ref{eq1.3})$ holds, then for any sufficiently small $\varepsilon>0$, one can find a Borel subset $W\subset\varSigma_{\!N}^+$
with $\mathds{P}_{\vec{\alpha}}(W)\ge1-\varepsilon$ and a constant $\bbdelta>0$ such that for each $\biota\in W$, the switching-control systems
\begin{equation*}
\dot{x}(t)=A_{\sigma_{\biota}(t)}x(t)+B_{\sigma_{\biota}(t)}(x(t))u(t),\quad x(0)\in\mathbb{R}^n\textrm{ and } t\in\mathbb{R}_+ \leqno{(\ref{eq1.1})_{\biota,\varepsilon}}
\end{equation*}
are globally exponentially stable, for any measurable external input $u(\bcdot)$ with $u(t)\in\mathbb{R}^{m_{\sigma_\biota(t)}}$ and $\|u(t)\|\le\bbdelta$.
\end{description}
\end{thm}

Although the statement (1) of Theorem~\ref{thm1.2} is a direct consequence of the statement (2) in the case $\varepsilon=0$, the separated formulations are convenient for our arguments later.

\begin{remark}\label{rem1.3}
Under the solvability property of $\mathcal{A}$, condition $(\star)$ is not only sufficient but also necessary for the statement $\mathrm{(1)}$ of Theorem~\ref{thm1.2};
see Theorem~\ref{thm4.2} below for the full details.
\end{remark}

\begin{remark}\label{rem1.4}
Our arguments presented here imply that replacing the solvability of $\mathcal{A}$ by a more general condition that $\mathcal{A}$ admits a simultaneous triangularization, the statements of Theorem~\ref{thm1.2} and Remark~\ref{rem1.3} still hold.
\end{remark}

Determining whether or not a set $\mathcal{A}$ of matrices admits a simultaneous triangularization is itself a long studied problem; for example, see \cite{Laffey78, AI, RR}.

\begin{remark}\label{rem1.5}
The statements of Theorem~\ref{thm1.2} are given in an almost sure sense. However, since an arbitrary point $\biota\in\varSigma_{\!N}^+$ is a density point of $\mathds{P}_{\vec{\alpha}}$, we can choose a sequence of stable switching sequence $\biota^{(\ell)}\in\varSigma_{\!N}^+$ such that $\biota^{(\ell)}\to\biota$ as $\ell\to+\infty$.
Secondly, it is a well known fact that the Hausdorff dimension of $\mathds{P}_{\vec{\alpha}}$ is equal to its entropy up to a constant multiplicator \cite{DHX08}. So, if the barycenter $(A_1+\dotsm+A_N)/N$ of the convex hull $\mathrm{co}\{A_1,\dotsc,A_N\}$ is stable, then the set
\begin{equation*}
\left\{\biota\in\varSigma_{\!N}^+\,|\,\dot{x}=A_{\sigma_\biota(t)}x\textrm{ is exponentially stable}\right\}
\end{equation*}
has the same Hausdorff dimension as the symbolic space $\varSigma_{\!N}^+$ under any standard metrics. Thirdly, from \cite{Dai10}, it follows that for every point $\biota$ in the basin of $\mathds{P}_{\vec{\alpha}}$, $\sigma_\biota(t)$ is a stable switching signal under the $(\star)$-condition.
\end{remark}

\begin{remark}\label{rem1.6}
For any $A\in\mathbb{R}^{n\times n}$, it is stable $\mathit{iff}$ each of its eigenvalues has a negative real part. So, if there exists a stable $\bbA$ in the convex hull
$\mathrm{co}\{A_1,\dotsc,A_N\}$ in $\mathbb{R}^{n\times n}$, then one always can find a positive probability vector $\vec{\alpha}=(\alpha_1+\dotsm+\alpha_N)\in\mathbb{R}^N$ such that the convex combination $\alpha_1A_1+\cdots+\alpha_NA_N$ is stable.
\end{remark}

Let us see our Theorem~\ref{thm1.2} from the viewpoint of approximations of solutions of the linear affine equation
\begin{equation}\label{eq1.9}
\dot{x}(t)=(\alpha_1A_1+\cdots+\alpha_NA_N)x(t),\quad x(0)\in\mathbb{R}^n\textrm{ and } t\in\mathbb{R}_+
\end{equation}
by solutions of the differential inclusion
\begin{equation}\label{eq1.10}
\dot{y}(t)\in\{A_1,\dotsc, A_N\}y(t),\quad y(0)\in\mathbb{R}^n\textrm{ and } t\in\mathbb{R}_+.
\end{equation}
There is a stabilizing switching strategy, proposed in \cite{WPD}, as follows: If condition $(\star)$ of Theorem~\ref{thm1.2} holds, i.e., (\ref{eq1.9}) is stable, then from the continuous version of the Filippov-Wa\v{z}ewski relaxation theorem (cf.~\cite{FR, ISW}), it easily follows that to any initial state $y(0)=\xi$, one can find an associated switching signal $\sigma_{\!\xi}(t)$ such that the solution $y(t,\xi,\sigma_{\!\xi})$, with $y(0)=\xi$, of the switched system
\begin{equation*}
\dot{y}(t)=A_{\sigma_{\!\xi}(t)}y(t),\quad y(0)\in\mathbb{R}^n\textrm{ and } t\in\mathbb{R}_+\leqno{(\ref{eq1.10})_\xi}
\end{equation*}
converges to zero exponentially fast; that is, it holds that
\begin{equation*}
\bbchi^+(\xi,\sigma_\xi,0)=\limsup_{t\to+\infty}\frac{1}{t}\log\|y(t,\xi,\sigma_{\!\xi})\|<0;
\end{equation*}
this phenomenon is called {\it pointwise exponentially stabilizable} in \cite{Sun04,Sun06}. Yet one cannot claim the global stability of the above switching system $(\ref{eq1.10})_\xi$ steered by such switching signal $\sigma_{\!\xi}(t)$! That is to say, different initial states $\xi$ may define different switching signals $\sigma_\xi(t)$ suggested by \cite{WPD}. In \cite{Tok87} also see \cite{Sun04}, based on the Baker-Campbell-Haudorff formula J. Tokarzewski found a periodically switched signal $\bbsigma(t)$ which enables the individual switched dynamical system
\begin{equation*}
\dot{x}(t)=A_{\bbsigma(t)}x(t),\quad x(0)\in\mathbb{R}^n\textrm{ and }t\in\mathbb{R}_+\leqno{(\ref{eq1.10})_{\bbsigma}}
\end{equation*}
exponentially stable. However, the period of $\bbsigma(t)$ is sufficiently small and essentially the above switched dynamical system $(\ref{eq1.10})_{\bbsigma}$ defined by the periodic $\bbsigma(t)$ is a small perturbation of the stable system (\ref{eq1.9}) there. Moreover, it is well known that the set of all periodically switched signals in $\varSigma_{\!N}^+$ is countable, $0$-Hausdorff-dimensional, and has only $\mathds{P}_{\vec{\alpha}}$-measure zero.

Clearly, under the additional solvability or simultaneous triangularization conditions, our statement (1) of Theorem~\ref{thm1.2} presented in this paper is much more stronger than those mentioned above. For the nonlinear case, the relaxation theorem requires an additional Lipschitz condition for the nonlinear controlled part. There is a similar comparison if we additionally assume the Lipschitz continuity of $B_i(x)$ with respect to $x\in\mathbb{R}^n$ for all indices $1\le i\le N$.

This paper is organized as follows. The rest Sections~\ref{sec2}, \ref{sec3}, \ref{sec4}, and \ref{sec5} of the paper are all devoted to proving Theorem~\ref{thm1.2}.
In Section~\ref{sec2}, we will provide an exponential stability criterion for a time-dependent continuous-time linear equation whose coefficient matrix is upper-triangular and complex, see Theorem~\ref{thm2.1} below. In Section~\ref{sec3}, we will introduce a
continuous-time symbolic semiflow by suspension of the classical one-sided Markovian shift transformation $(\varSigma_{\!N}^+,\theta)$. Then, we can think of our switching dynamical systems as skew-product semiflows driven by the continuous-time symbolic semiflow. Borrowing the symbolic semiflow and Lie's theorem of triangularization, we can apply ergodic theory to proving Theorem~\ref{thm1.2}. We will prove the statements (1) of Theorem~\ref{thm1.2} and Remark~\ref{rem1.3} in Section~\ref{sec4} and the statement (2) of Theorem~\ref{thm1.2} in Section~\ref{sec5}. Since in our context $B_i(x)$ may lack the higher regularity in $x$, the classical Lyapunov stability theorems cannot work here for proving the statement (2) of Theorem~\ref{thm1.2}. So, we will employ in Section~\ref{sec5} a new tool\,---\,Liao-type exponents, first introduced in \cite{DHX10} and then perfected by the recent work \cite{DHX11}.
In fact, there we will prove a more general result Proposition~\ref{prop5.2}.

Finally we are going to conclude this introductory section with a question for further study.

\begin{que}
Let $\mathcal{A}=\{A_1,\dotsc,A_N\}\subset\mathbb{R}^{n\times n}$ and $\vec{\alpha}=(\alpha_1,\dotsc,\alpha_N)$ be a positive probability vector. If the $(\star)$-condition is satisfied then, do the statements of Theorem~\ref{thm1.2} still hold without the solvability condition of $\mathcal{A}$?
\end{que}


\section{Lyapunov exponents for linear differential equations}\label{sec2}%

In this section, we will consider a general linear differential equation
\begin{equation}\label{eq2.1}
\dot{x}(t)=C(t)x(t),\quad x(0)\in\mathbbm{C}^n\textrm{ and } t\in\mathbb{R}_+,
\end{equation}
where $C(t)=\left[c^{ij}(t)\right]_{1\le i,j\le n}\in\mathbbm{C}^{n\times n}$, with complex elements, is a Borel-measurable matrix-valued function of the time-variable $t$ on $\mathbb{R}_+$.
We assume that the matrix function $C(t)$ is bounded on $\mathbb{R}_+$; that is,
\begin{equation}\label{eq2.2}
{\sup}_{t\in\mathbb{R}_+}\|C(t)\|<\infty.
\end{equation}
It is well known that under condition (\ref{eq2.2}), for every initial state $x_0\in\mathbbm{C}^n$ there exists a unique solution of (\ref{eq2.1}), written as
$x(t)=\varPhi(t) x_0$, which is defined on $\mathbb{R}_+$ such that $x(0)=x_0$. This implies that $\varPhi(t)$ is the principal matrix of (\ref{eq2.1}); that is to say, \begin{equation*}
\varPhi(0)=\mathrm{Id}_{\mathbbm{C}^n}\textrm{ (the unit matrix)}\quad \textrm{and}\quad\dot{\varPhi}(t)=C(t)\varPhi(t)\textrm{ for \textsf{Leb}-a.s. }t\in\mathbb{R}_+.
\end{equation*}
Here \textsf{Leb} denotes the usual Lebesgue measure of $\mathbb{R}_+$. Then, the number
\begin{equation}\label{eq2.3}
\bbchi^+:=\limsup_{t\to+\infty}\frac{1}{t}\log\|\varPhi(t)\|\in\mathbb{R}\cup\{-\infty\}
\end{equation}
is called the (maximal) {\it Lyapunov exponent} of (\ref{eq2.1}). Clearly, for every nonzero initial state $x_0\in\mathbbm{C}^n$, its Lyapunov exponent
\begin{equation}\label{eq2.4}
\bbchi^+(x_0):=\limsup_{t\to+\infty}\frac{1}{t}\log\|\varPhi(t)x_0\|\le\bbchi^+.
\end{equation}
If $\bbchi^+<0$, then we call (\ref{eq2.1}) {\it exponentially stable}. According to the classical Lyapunov theory, see e.g.~\cite{Lya}, $\bbchi^+(x_0)$ can take at most $n$ distinct values for all $x_0\in\mathbbm{C}^n\setminus\{\mathbf{0}\}$. The basic question is: \textit{Does there hold $\max\{\bbchi^+(x_0)\,|\,x_0\in\mathbbm{C}^n\setminus\{\mathbf{0}\}\}=\bbchi^+$?} If $C(t)\equiv C(0)$ for all $t\in\mathbb{R}_+$ then the answer is positive. In general, we will see this is still true from Theorem~\ref{thm2.1} below.

For the system (\ref{eq2.1}), a very interesting fact is that, generally speaking, the stability of the time-invariant systems
\begin{equation*}
\dot{x}(t)=C(T)x(t),\quad x(0)\in\mathbbm{C}^n\textrm{ and } t\in\mathbb{R}_+,
\end{equation*}
for every $T>0$, cannot imply the stability of (\ref{eq2.1}); this point is well-illustrated by the
classical Marcus-Yamabe example. Consider the linear periodic differential equation
\begin{equation*}
\dot{x}(t)=\left[\begin{array}{ll}-2+2\cos^2t&1-\sin2t\\
-1-\sin2t&-2+2\sin^2t\end{array}\right]x(t)=A(t)x(t),\quad x(0)\in\mathbb{R}^2\textrm{ and } t\in\mathbb{R};
\end{equation*}
one checks that, for each $T\in\mathbb{R}$, $A(T)$ admits $\bblambda=-1$ as an eigenvalue of algebraic multiplicity $2$ and so $\dot{x}(t)=A(T)x(t)$ are stable for all $T$; however, the differential equation admits the exponentially unstable solution $x(t)=(-e^t\cos t, e^t\sin t)^\mathrm{T}$, where $^\mathrm{T}$ means the transpose operation of a square matrix or a column/row vector.

Particularly, we will be interested in the upper-triangular equations. For this, our result is the following, which is implicitly contained in the proof of the widely known Perron-Lyapunov regularity theorem~\cite{Lya}.

\begin{thm}\label{thm2.1}
If $C(t)=\left[c^{ij}(t)\right]\in \mathbbm{C}^{n\times n}$ is upper-triangular, i.e., $c^{ij}(t)\equiv0$ for all $n\ge i>j\ge1$, for $t\in\mathbb{R}_+$, then it holds that
\begin{equation*}
\bbchi^+=\max\{\vartheta_i\,|\,i=1,\dotsc,n\},
\end{equation*}
where
\begin{equation*}
\vartheta_i:=\limsup_{T\to+\infty}\frac{1}{T}\int_0^T\mathfrak{Re}(c^{ii}(t))\,\mathrm{d}t\quad \forall i=1,\dotsc,n.
\end{equation*}
Thus, (\ref{eq2.1}) is exponentially stable if and only if $\vartheta_i<0$ for all $i=1,\dotsc,n$.
\end{thm}

Here $\mathfrak{Re}(c)$ means the real part of a complex number $c\in\mathbbm{C}$. This result shows that if the ``time-average" time-invariant system
\begin{equation*}
\dot{x}(t)=\widetilde{C}x(t),\quad x\in\mathbbm{C}^n\textrm{ and }t\in\mathbb{R}_+\qquad \textrm{where }\widetilde{C}=\limsup_{T\to+\infty}\frac{1}{T}\int_0^TC(t)\,\mathrm{d}t,
\end{equation*}
is exponentially stable, then so is (\ref{eq2.1}). Our condition $(\star)$, formulated in Theorem~\ref{thm1.2}, is essentially an other kind of average\,---\,``spatial average". This point will be well illustrated in Section~\ref{sec4}.

\begin{proof}
We only consider here the simple, but nontrivial, case of the order $n=2$; the general case can be similarly proved as this.

Now, we define the matrix-valued function $\varPhi_\varDelta(t)=[\phi_{ij}(t)]\in\mathbbm{C}^{2\times 2}$ for $t\in\mathbb{R}_+$ and $0\le\varDelta\le+\infty$ as follows:
\begin{align*}
&\phi_{11}(t)=\exp\left(\int_0^tc^{11}(\tau)\mathrm{d}\tau\right),&&\phi_{21}(t)=0,\\
\intertext{and}&\phi_{12}(t)={\int}_{\!\!\varDelta}^t c^{12}(s)\phi_{22}(s)\exp\left(\int_s^tc^{11}(\tau)\mathrm{d}\tau\right)\mathrm{d}s, &&\phi_{22}(t)=\exp\left(\int_0^tc^{22}(\tau)\mathrm{d}\tau\right).
\end{align*}
It is easily seen that for any constant $\varDelta\in[0,+\infty]$, the columns of $\varPhi_\varDelta(t)$ form a basis of solutions of (\ref{eq2.1}) in the case $n=2$. So, for any $\varDelta\ge0$, we have by (\ref{eq2.3})
$$\bbchi^+=\limsup_{t\to+\infty}\frac{1}{t}\log\|\varPhi_\varDelta(t)\|.$$
Therefore, we need to prove only that $\bbchi^+\le\max\{\vartheta_1,\vartheta_2\}$.
We next consider the two solutions of (\ref{eq2.1}) described in $\varPhi_\varDelta(t)$ and will show that
\begin{equation*}
\pmb{\phi}_1(t):=\left[\begin{matrix}\phi_{11}(t)\\ \phi_{21}(t)\end{matrix}\right]\quad\textrm{and}\quad\pmb{\phi}_2(t):=\left[\begin{matrix}\phi_{12}(t)\\ \phi_{22}(t)\end{matrix}\right]
\end{equation*}
satisfy, respectively, that
\begin{align*}
\bbchi^+(\pmb{\phi}_{1})&:=\limsup_{t\to+\infty}\frac{1}{t}\log\|\pmb{\phi}_{1}(t)\|=\vartheta_1\\
\intertext{and}\bbchi^+(\pmb{\phi}_{2})&:=\limsup_{t\to+\infty}\frac{1}{t}\log\|\pmb{\phi}_{2}(t)\|=\vartheta_2
\end{align*}
for some choice of the constant $\varDelta$.

In fact, from $\vartheta_i=\limsup_{t\to+\infty}\frac{1}{t}\int_0^t\mathfrak{Re}(c^{ii}(\tau))\,\mathrm{d}\tau$ for $i=1,2$, it easily follows that
\begin{equation*}
\vartheta_1=\limsup_{t\to+\infty}\frac{1}{t}\log\|\phi_{11}(t)\|\quad\textrm{and}\quad\vartheta_2=\limsup_{t\to+\infty}\frac{1}{t}\log\|\phi_{22}(t)\|.
\end{equation*}
Then for any constant $\varepsilon>0$, one can find a constant $K_{\varepsilon}>0$ such that
\begin{equation*}
\frac{\|\phi_{22}(t)\|}{\|\phi_{11}(t)\|}\le K_{\varepsilon}e^{(\vartheta_2-\vartheta_1+\varepsilon)t}\quad\forall t>0.
\end{equation*}
Now, we need only to estimate
$\bbchi^+(\phi_{12}):=\limsup_{t\to+\infty}\frac{1}{t}\log\|\phi_{12}(t)\|$.
For that, observe that for each $\varepsilon>0$, we have
\begin{align*}
\bbchi^+(\phi_{12})&\le\limsup_{t\to+\infty}\frac{1}{t}\left\{\log\|e^{\int_0^tc^{11}(\tau)\mathrm{d}\tau}\|+\log\|\int_{\!\!\varDelta}^tc^{12}(s)\phi_{22}(s)e^{\int_s^0c^{11}(\tau)\mathrm{d}\tau}\mathrm{d}s\|\right\}\\
&\le\limsup_{t\to+\infty}\frac{1}{t}\left\{\log\|\phi_{11}(t)\|+\log\|\int_{\!\!\varDelta}^tK\|\phi_{22}(s)\|/\|\phi_{11}(s)\|\mathrm{d}s\|\right\}\\
&\le\vartheta_1+\limsup_{t\to+\infty}\frac{1}{t}\log\|\int_{\!\!\varDelta}^tKK_{\varepsilon}e^{(\vartheta_2-\vartheta_1+\varepsilon)s}\mathrm{d}s\|.
\end{align*}
Here we exploit the fact that $\|c^{12}(t)\|\le K$ for some constant $K$ by condition (\ref{eq2.2}). We set $\varDelta=0$ if $\vartheta_2-\vartheta_1\ge 0$, and $\varDelta=+\infty$ if $\vartheta_2-\vartheta_1<0$. Then, for every $\varepsilon>0$ so small that $\vartheta_2-\vartheta_1+\varepsilon<0$ if $\vartheta_2-\vartheta_1<0$, we can obtain that
\begin{align*}
&\bbchi^+(\phi_{12})\le\vartheta_1+\limsup_{t\to+\infty}\frac{1}{t}\log\frac{KK_{\varepsilon}[e^{(\vartheta_2-\vartheta_1+\varepsilon)t}-1]}{\vartheta_2-\vartheta_1+\varepsilon}
&&\textrm{if }\vartheta_2-\vartheta_1\ge0,\\
\intertext{and}
&\bbchi^+(\phi_{12})\le\vartheta_1+\limsup_{t\to+\infty}\frac{1}{t}\log\frac{KK_{\varepsilon}e^{(\vartheta_2-\vartheta_1+\varepsilon)t}}{\vartheta_2-\vartheta_1+\varepsilon}
&&\textrm{if }\vartheta_2-\vartheta_1<0.
\end{align*}
Thus, no matter what $\varDelta$ is equal to, it always holds that
\begin{equation*}
\bbchi^+(\phi_{12})\le\vartheta_1+(\vartheta_2-\vartheta_1+\varepsilon)=\vartheta_2+\varepsilon.
\end{equation*}
Since $\varepsilon>0$ is arbitrary, there follows that $\bbchi^+(\phi_{12})\le\vartheta_2$.
So,
$$\bbchi^+(\pmb{\phi}_1)=\vartheta_1\quad \textrm{and}\quad \bbchi^+(\pmb{\phi}_2)=\vartheta_2,$$
as desired. Thus, $\bbchi^+\le\max\{\vartheta_1,\vartheta_2\}$ and then $\bbchi^+=\max\{\vartheta_1,\vartheta_2\}$.
This completes the proof of the theorem.
\end{proof}

It is interesting to note that based on Theorem~\ref{thm2.1}, the Marcus-Yamabe phenomenon still can occur even if in the upper-triangular case. This point may be illustrated by the following linear upper-triangular equation:
\begin{equation*}
\dot{x}(t)=\left[\begin{matrix}-\frac{1}{1+t}& a^{12}(t)\\
0&-\frac{1}{1+t}\end{matrix}\right]x(t)=A(t)x(t),\quad t\in\mathbb{R}_+\textrm{ and } x(0)\in\mathbb{R}^2;
\end{equation*}
one easily checks that, for each $T\in\mathbb{R}_+$, $A(T)$ admits $\bblambda=-\frac{1}{1+T}<0$ as an eigenvalue of algebraic multiplicity $2$; however, this differential equation has the Lyapunov exponent $\bbchi^+=0$ from Theorem~\ref{thm2.1}, not less than $0$; so not exponentially stable.

\section{Symbolic semiflow}\label{sec3}
In this section, we will embed, in the manner of suspension, the one-sided Markov shift transformation $\theta\colon\varSigma_{\!N}^+\rightarrow\varSigma_{\!N}^+$ defined as (\ref{eq1.5}) in Section~\ref{sec1} into a continuous-time semiflow. The latter is important for us, which plays the role of the driving system for the switching dynamical system that we are interested to here.
\subsection{Suspension of $(\varSigma_{\!N}^+,\theta,\mathds{P}_{\vec{\alpha}})$}\label{sec3.1}%
First, we can define simply a continuous-time semiflow on the Cartesian product topological space $\varSigma_{\!N}^+\times\mathbb{R}_+$ as follows:
\begin{equation*}
\psi\colon[0,\infty)\times(\varSigma_{\!N}^+\times\mathbb{R}_+)\rightarrow\varSigma_{\!N}^+\times\mathbb{R}_+;\quad (t,(\biota,\tau))\mapsto(\biota, \tau+t).
\end{equation*}
We note here that $\mathbb{R}_+=(0,\infty)$. Next, we introduce an equivalent relationship, written as $``\sim"$, for $\varSigma_{\!N}^+\times\mathbb{R}_+$ as follows:
\begin{equation*}
(\xi,t)\sim(\biota, t^\prime)\textrm{ for }t, t^\prime>0 \textrm{ and }\xi, \biota\in\varSigma_{\!N}^+\Leftrightarrow t-t^\prime\in\mathbb{Z}\textrm{ and } \begin{cases}
\xi=\theta^k(\biota) \textrm{ if }k=t^\prime-t\ge0,\\
\biota=\theta^k(\xi) \textrm{ if }k=t-t^\prime\ge0.
\end{cases}
\end{equation*}
We then define the suspension space as the quotient space
\begin{equation}\label{eq3.1}
\mathrm{S}(\varSigma_{\!N}^+)=\varSigma_{\!N}^+\times\mathbb{R}_+/\sim.
\end{equation}
Clearly, there is no loss of generality in assuming $\mathrm{S}(\varSigma_{\!N}^+)=\{(\biota,t)\,|\,\biota\in\varSigma_{\!N}^+\textrm{ and }0\le t\le 1\}$ with $(\biota,1)$ and $(\theta(\biota),0)$ identified, namely
\begin{equation*}
\mathrm{S}(\varSigma_{\!N}^+)=\varSigma_{\!N}^+\times\mathbb{I}/\sim  \leqno{(\ref{eq3.1})^\prime}
\end{equation*}
where $\mathbb{I}=[0,1]$ is the unit interval. Write the elements of $\mathrm{S}(\varSigma_{\!N}^+)$ or $\sim$-classes of $\varSigma_{\!N}^+\times\mathbb{R}_+$ as $[\biota,\tau]$ from here on. So, $\mathrm{S}(\varSigma_{\!N}^+)$ is a compact metrizable topological space, which is indeed homeomorphic to the compact product space $\varSigma_{\!N}^+\times\mathds{T}^1$, where $\mathds{T}^1$ is the unit circle. Observe that $\psi$ can further induce a continuous-time semiflow, called the {\it symbolic semiflow}, on $\mathrm{S}(\varSigma_{\!N}^+)$ as follows:
\begin{equation}\label{eq3.2}
\varTheta\colon[0,\infty)\times\mathrm{S}(\varSigma_{\!N}^+)\rightarrow\mathrm{S}(\varSigma_{\!N}^+);\quad (t,[\biota,\tau])\mapsto[\biota,\tau+t].
\end{equation}
For example, $\varTheta(\frac{3}{2}, [\biota,\frac{1}{2}])=[\biota,2]=[\theta^2(\biota),0]$ for any $\biota\in\varSigma_{\!N}^+$.

If we identify $\varSigma_{\!N}^+$ with the $0$-section $[\varSigma_{\!N}^+,0]\subset\mathrm{S}(\varSigma_{\!N}^+)$, then $\theta$ may be thought of as the $1$-time map $\varTheta(1,\cdot)$ restricted to $[\varSigma_{\!N}^+,0]$, noting that $[\varSigma_{\!N}^+,t+k]=[\varSigma_{\!N}^+,t]$ for any $t\in\mathbb{R}_+$ and $k\in\mathbb{N}$ since $\theta(\varSigma_{\!N}^+)=\varSigma_{\!N}^+$.

Given any positive probability vector $\vec{\alpha}=(\alpha_1,\dotsc,\alpha_N)\in\mathbb{R}^N$ and let $\textsf{Leb}$ denote the standard Lebesgue measure on the unit interval $\mathbb{I}$ with $\textsf{Leb}(\mathbb{I})=1$. Then, we naturally define the product probability on the suspension space $\mathrm{S}(\varSigma_{\!N}^+)$
\begin{equation}\label{eq3.3}
\mathcal{P}_{\vec{\alpha}}=\mathds{P}_{\vec{\alpha}}\otimes\textsf{Leb}
\end{equation}
in this way: for any Borel sets $B_1\subset\varSigma_{\!N}^+$ and $B_2\subset \mathbb{I}$, we have
\begin{equation*}
\mathcal{P}_{\vec{\alpha}}([B_1,B_2])=\mathds{P}_{\vec{\alpha}}(B_1)\cdot\textsf{Leb}(B_2),\quad \textrm{where }[B_1,B_2]:=\{[\biota,\tau]\colon\biota\in B_1,\; \tau\in B_2\}.
\end{equation*}
Equivalently, we have
\begin{equation*}
\int_{\mathrm{S}(\varSigma_{\!N}^+)}\varphi([\biota,\tau])\,\mathrm{d}\mathcal{P}_{\vec{\alpha}}([\biota,\tau])=\int_{\varSigma_{\!N}^+}\left\{\int_\mathbb{I}\varphi([\biota,\tau])\,\mathrm{d}\textsf{Leb}(\tau)\right\}\mathrm{d}\mathds{P}_{\vec{\alpha}}(\biota)
\end{equation*}
for any $\varphi\in \mathrm{C}^0(\mathrm{S}(\varSigma_{\!N}^+),\mathbb{R})$ (the space of continuous functions on $\mathrm{S}(\varSigma_{\!N}^+)$), where $\mathds{P}_{\vec{\alpha}}$ is defined in the same manner as in (\ref{eq1.7}) and $\mathrm{d}\textsf{Leb}(\tau)=\mathrm{d}\tau$.

\begin{lem}\label{lem3.1}
For any probability vector $\vec{\alpha}=(\alpha_1,\dotsc,\alpha_N)>0$, the corresponding $\mathcal{P}_{\vec{\alpha}}$ is invariant and ergodic for the symbolic semiflow $\varTheta\colon[0,\infty)\times\mathrm{S}(\varSigma_{\!N}^+)\rightarrow\mathrm{S}(\varSigma_{\!N}^+)$; that is to say, $(\mathrm{S}(\varSigma_{\!N}^+),\varTheta,\mathcal{P}_{\vec{\alpha}})$ is an ergodic semiflow.
\end{lem}

We notice here that the statement of Lemma~\ref{lem3.1} still holds if $\vec{\alpha}$ is only a probability vector, not necessarily positive, by ignoring the letters $k$ whenever $\alpha_k=0$.

\begin{proof}
We notice that the $\varTheta$-invariance of $\mathcal{P}_{\vec{\alpha}}$ is standard. In fact, for any $\varphi$ in $\mathrm{C}^0(\mathrm{S}(\varSigma_{\!N}^+),\mathbb{R})$ and $0<t<1$, we have
\begin{align*}
\int_{\mathrm{S}(\varSigma_{\!N}^+)}\varphi(\varTheta(t,\bcdot))\,\mathrm{d}\mathcal{P}_{\vec{\alpha}}&=\int_{\varSigma_{\!N}^+}\left\{\int_0^1\varphi([\biota,\tau+t])\,\mathrm{d}\textsf{Leb}(\tau)\right\}\mathrm{d}\mathds{P}_{\vec{\alpha}}(\biota)\\
&=\int_{\varSigma_{\!N}^+}\left\{\int_t^{1+t}\varphi([\biota,\tau])\,\mathrm{d}\textsf{Leb}(\tau)\right\}\mathrm{d}\mathds{P}_{\vec{\alpha}}(\biota)\\
&=\int_{\varSigma_{\!N}^+}\left\{\int_t^1\varphi([\biota,\bcdot])\,\mathrm{d}\tau+\int_0^t\varphi([\theta(\biota),\bcdot])\,\mathrm{d}\tau\right\}\mathrm{d}\mathds{P}_{\vec{\alpha}}(\biota)\\
&=\int_{\varSigma_{\!N}^+}\int_t^1\varphi([\biota,\tau])\,\mathrm{d}\tau\mathrm{d}\mathds{P}_{\vec{\alpha}}(\biota)
+\int_0^t\int_{\varSigma_{\!N}^+}\varphi([\theta(\biota),\tau])\,\mathrm{d}\mathds{P}_{\vec{\alpha}}(\biota)\mathrm{d}\tau\\
&=\int_{\varSigma_{\!N}^+}\int_t^1\varphi([\biota,\tau])\,\mathrm{d}\tau\mathrm{d}\mathds{P}_{\vec{\alpha}}(\biota)
+\int_0^t\int_{\varSigma_{\!N}^+}\varphi([\biota,\tau])\,\mathrm{d}\mathds{P}_{\vec{\alpha}}(\biota)\mathrm{d}\tau\\
&=\int_{\mathrm{S}(\varSigma_{\!N}^+)}\varphi([\biota,\tau])\,\mathrm{d}\mathcal{P}_{\vec{\alpha}}([\biota,\tau]);
\end{align*}
similarly, for any $t\in\mathbb{N}$ we have
\begin{equation*}
\int_{\mathrm{S}(\varSigma_{\!N}^+)}\varphi(\varTheta(t,[\biota,\tau]))\,\mathrm{d}\mathcal{P}_{\vec{\alpha}}([\biota,\tau])=\int_{\mathrm{S}(\varSigma_{\!N}^+)}\varphi([\biota,\tau])\,\mathrm{d}\mathcal{P}_{\vec{\alpha}}([\biota,\tau]).
\end{equation*}
This implies that $\mathcal{P}_{\vec{\alpha}}$ is $\varTheta$-invariant.
So, it remains to prove only the $\varTheta$-ergodicity of $\mathcal{P}_{\vec{\alpha}}$. For that, given any $0<\rho<1$, let
\begin{align*}
&|\varphi|_\rho=\sup_{\ell\in\mathbb{N}}\left\{\frac{\sup\{|\varphi(\biota)-\varphi(\biota^\prime)|\colon\biota_j=\biota_j^\prime,\ 1\le j\le \ell\}}{\rho^\ell}\right\}\quad \forall\varphi\in\mathrm{C}^0(\varSigma_{\!N}^+,\mathbb{R})\\
\intertext{and}&\mathscr{H}_{\!\rho}(\varSigma_{\!N}^+)=\left\{\varphi\in\mathrm{C}^0(\varSigma_{\!N}^+,\mathbb{R})\colon|\varphi|_\rho<\infty\right\}.
\end{align*}
Elements of $\mathscr{H}_{\!\rho}(\varSigma_{\!N}^+)$ are referred to as H\"{o}lder continuous functions. Let $\mathcal{L}_\mathbb{R}^\infty(\mathrm{S}(\varSigma_{\!N}^+))$ be the set of bounded, Borel measurable, real-valued functions on $\mathrm{S}(\varSigma_{\!N}^+)$. Define $\mathscr{F}_{\!\rho}(\mathrm{S}(\varSigma_{\!N}^+))$ as the set of the functions $G\in\mathcal{L}_\mathbb{R}^\infty(\mathrm{S}(\varSigma_{\!N}^+))$ satisfying
\begin{equation*}
g(\biota)=\int_0^1G([\biota,\tau])\,\mathrm{d}\textsf{Leb}(\tau)\in\mathscr{H}_{\!\rho}(\varSigma_{\!N}^+).
\end{equation*}
Let $\mathcal{M}_{\mathrm{inv}}(\mathrm{S}(\varSigma_{\!N}^+),\varTheta)$ and $\mathcal{M}_{\mathrm{inv}}(\varSigma_{\!N}^+,\theta)$ be the sets of $\varTheta$-invariant Borel probability measures on $\mathrm{S}(\varSigma_{\!N}^+)$ and $\theta$-invariant Borel probability measures on $\varSigma_{\!N}^+$, respectively. They are both nonempty compact convex sets from ergodic theory, see \cite{NS, Walters82}. Then from \cite[Section~5]{Lalley} (also cf.~\cite{Abramov, Gurevich}), it follows that there is a $1$-to-$1$ correspondence between $\mathcal{M}_{\mathrm{inv}}(\mathrm{S}(\varSigma_{\!N}^+),\varTheta)$ and $\mathcal{M}_{\mathrm{inv}}(\varSigma_{\!N}^+,\theta)$, given by
\begin{equation*}
\mathcal{M}_{\mathrm{inv}}(\mathrm{S}(\varSigma_{\!N}^+),\varTheta)\ni\bar{\mu}\overset{\pi_{\!*}}{\longleftrightarrow}\mu\in\mathcal{M}_{\mathrm{inv}}(\varSigma_{\!N}^+,\theta)
\end{equation*}
if and only if
\begin{equation*}
\int_{\mathrm{S}(\varSigma_{\!N}^+)}G\, \mathrm{d}\bar{\mu}=\int_{\varSigma_{\!N}^+}\left\{\int_0^1G([\biota,\tau])\,\mathrm{d}\tau\right\}\mathrm{d}\mu(\biota)\quad\forall G\in\mathscr{F}_{\!\rho}(\mathrm{S}(\varSigma_{\!N}^+)).
\end{equation*}

Now, if $\mathcal{P}_{\vec{\alpha}}$ would not be $\varTheta$-ergodic, then \cite[Theorem~6.10.(iii)]{Walters82} follows that there are $0<\lambda<1$ and $\bar{\mu}_1,\bar{\mu}_2\in\mathcal{M}_{\mathrm{inv}}(\mathrm{S}(\varSigma_{\!N}^+),\varTheta)$ such that
\begin{equation*}
\mathcal{P}_{\vec{\alpha}}=\lambda\bar{\mu}_1+(1-\lambda)\bar{\mu}_2\quad \textrm{and}\quad \bar{\mu}_1\not=\bar{\mu}_2.
\end{equation*}
Let $\bar{\mu}_1\overset{\pi_{\!*}}{\longleftrightarrow}\mu_1\in\mathcal{M}_{\mathrm{inv}}(\varSigma_{\!N}^+,\theta)$ and $\bar{\mu}_2\overset{\pi_{\!*}}{\longleftrightarrow}\mu_2\in\mathcal{M}_{\mathrm{inv}}(\varSigma_{\!N}^+,\theta)$. Then, by the definitions of $\pi_*$ and $\mathds{P}_{\vec{\alpha}}$, we have
\begin{equation*}
\mu_1\not=\mu_2,\quad \mathcal{P}_{\vec{\alpha}}\overset{\pi_{\!*}}{\longleftrightarrow}\lambda\mu_1+(1-\lambda)\mu_2\quad \textrm{and}\quad\mathcal{P}_{\vec{\alpha}}\overset{\pi_{\!*}}{\longleftrightarrow}\mathds{P}_{\vec{\alpha}}.
\end{equation*}
Thus, $\mathds{P}_{\vec{\alpha}}=\lambda\mu_1+(1-\lambda)\mu_2$ and moreover $\mathds{P}_{\vec{\alpha}}$ is not ergodic for $\theta$, a contradiction to the classical theory of Markov chains.

This completes the proof of the lemma.
\end{proof}

For a general discrete-time ergodic semidynamics $(X,\theta,\mu)$, one can similarly consider its continuous-time ergodic suspension semiflow $(\mathrm{S}(X),\varTheta,\mathcal{P}_\mu)$.

\subsection{Linear skew-product semiflow}\label{sec3.2}%
We now see the switching dynamics from the point of view of skew-product semiflow.

Let $\mathcal{A}=\{A_1,\dotsc,A_N\}\subset\mathbb{R}^{n\times n}$ be an arbitrarily given collection of real matrices. Then, it gives rise to the switching dynamical systems
\begin{equation}\label{eq3.4}
\dot{x}(t)=A_{\sigma_{\biota}(t)}x(t),\quad x(0)\in\mathbb{R}^n, t\in\mathbb{R}_+, \textrm{ and }\biota\in\varSigma_{\!N}^+,
\end{equation}
where the switching signals $\sigma_{\biota}\colon\mathbb{R}_+\rightarrow\{1,\dotsc,N\}$ are defined in the same way as in (\ref{eq1.8}) in Section~\ref{sec1}.
For any initial state $x(0)=x_0\in\mathbb{R}^n$, the solution to (\ref{eq3.4}) steered by $\biota\in\varSigma_{\!N}^+$ can be expressed as follows:
\begin{equation*}
x(t,x_0)=\begin{cases}
e^{tA_{\biota_1}} x_0& \textrm{if }0<t\le1,\\
e^{(t-k)A_{\biota_{k+1}}} e^{A_{\biota_k}}\cdots e^{A_{\biota_1}} x_0&\textrm{if }k<t\le k+1\textrm{ and }k\in\mathbb{N}.
\end{cases}
\end{equation*}

Now, on the suspension space $\mathrm{S}(\varSigma_{\!N}^+)$ defined in the manner as in Section~\ref{sec3.1}, we define a random matrix as follows:
\begin{equation}\label{eq3.5}
\mathds{A}_{\bcdot}\colon\mathrm{S}(\varSigma_{\!N}^+)\rightarrow\mathbb{R}^{n\times n};\quad [\biota,\tau]\mapsto A_{\sigma_{\biota}(\tau)}\;\forall \iota\in\varSigma_{\!N}^+\textrm{ and }0<\tau\le1.
\end{equation}
Because $\mathcal{P}_{\vec{\alpha}}([\varSigma_{\!N}^+,0])=0$, we do not care the evaluation of $\mathds{A}_{\bcdot}$ at the section $\varSigma_{\!N}^+\times\{0\}$.
Since $\mathds{A}_{[\biota,\tau+k]}=\mathds{A}_{[\theta^k(\biota),\tau]}$ by $\sigma_\biota(\tau+k)=\sigma_{\theta^k(\biota)}(\tau)$ for every $[\biota,\tau]\in\mathrm{S}(\varSigma_{\!N}^+)$ and $k\in\mathbb{N}$, this makes sense and satisfies
\begin{equation*}
\mathds{A}_{\varTheta(t,[\biota,\tau])}=\mathds{A}_{[\biota,\tau+t]}=A_{\sigma_\biota(\tau+t)}\quad\forall\tau\in\mathbb{R}_+\textrm{ and }0\le t<\infty.\leqno{(\ref{eq3.5})^\prime}
\end{equation*}
Then, we can obtain, from the random matrix $\mathds{A}_{\bcdot}$, a linear skew-product dynamical system
\begin{equation}\label{eq3.6}
\dot{x}(t)=\mathds{A}_{\varTheta(t,[\biota,\tau])}x(t),\quad 0\le t<\infty, x(0)\in\mathbb{R}^n,\textrm{ and }[\biota,\tau]\in\mathrm{S}(\varSigma_{\!N}^+)
\end{equation}
driven by the symbolic semiflow $\varTheta\colon[0,\infty)\times\mathrm{S}(\varSigma_{\!N}^+)\rightarrow\mathrm{S}(\varSigma_{\!N}^+)$.

Let $\{\varPhi_{[\biota,\tau]}(t)\}_{t\in\mathbb{R}_+}$ denote the principal matrix of (\ref{eq3.6}); that is to say, it satisfies $\dot{\varPhi}_{[\biota,\tau]}(t)=\mathds{A}_{\varTheta(t,[\biota,\tau])}\varPhi_{[\biota,\tau]}(t)$ at \textsf{Leb}-a.s. $t\in\mathbb{R}_+$ with $\varPhi_{[\biota,\tau]}(0)=\mathrm{Id}_{\mathbb{R}^n}$ the unit matrix. We can check the cocycle property
\begin{equation}\label{eq3.7}
\varPhi_{[\biota,\tau]}(t_1+t_2)=\varPhi_{\varTheta(t_1,[\biota,\tau])}(t_2)\varPhi_{[\biota,\tau]}(t_1)\quad \forall t_1,t_2\in\mathbb{R}_+.
\end{equation}
In fact, we need only to notice that for each $\biota\in\varSigma_{\!N}^+$, for any $\ell\in\mathbb{Z}_+$,
\begin{align*}
\varPhi_{[\biota,\ell]}(t)&=\begin{cases}e^{tA_{\biota_{1+\ell}}}& \textrm{if }0<t\le1,\\ e^{(t-k)A_{\biota_{k+\ell+1}}} e^{A_{\biota_{k+\ell}}}\cdots e^{A_{\biota_{1+\ell}}}&\textrm{if }k<t\le k+1\textrm{ and } k\in\mathbb{N};\end{cases}\\
\intertext{and for any $0<\tau<1$}
\varPhi_{[\biota,\tau]}(t)&=
\begin{cases}
e^{tA_{\biota_1}}& \textrm{if }0<t\le1-\tau,\\
e^{[t-(k-\tau)]A_{\biota_{k+1}}} e^{A_{\biota_k}}\cdots e^{(1-\tau)A_{\biota_1}}& \textrm{if }k-\tau<t\le k+1-\tau\textrm{ and }k\in\mathbb{N}.
\end{cases}
\end{align*}
This implies the cocycle property (\ref{eq3.7}).

From the above arguments, we can easily obtain the following result.

\begin{lem}\label{lem3.2}
Let $$\mathfrak{S}_{\theta,\mathcal{A}}\colon[0,\infty)\times\mathrm{S}(\varSigma_{\!N}^+)\times\mathbb{R}^n\rightarrow\mathrm{S}(\varSigma_{\!N}^+)\times\mathbb{R}^n$$
be defined in the following way:
\begin{equation*}
(t,([\biota,\tau],x))\mapsto([\biota,\tau+t],\varPhi_{[\biota,\tau]}(t) x)\quad\forall t\in[0,\infty)\textrm{ and } ([\biota,\tau],x)\in\mathrm{S}(\varSigma_{\!N}^+)\times\mathbb{R}^n.
\end{equation*}
Then, $\mathfrak{S}_{\theta,\mathcal{A}}$ is a continuous-time semiflow; that is to say,
\begin{align*}
&\mathfrak{S}_{\theta,\mathcal{A}}(0,([\biota,\tau],x))=([\biota,\tau],x),\\
&\mathfrak{S}_{\theta,\mathcal{A}}(t_1+t_2,([\biota,\tau],x))=\mathfrak{S}_{\theta,\mathcal{A}}\left(t_2,\mathfrak{S}_{\theta,\mathcal{A}}(t_1,([\biota,\tau],x))\right),\\
\intertext{and}&\mathfrak{S}_{\theta,\mathcal{A}}(t,([\biota,\tau],x))\textrm{ is jointly continuous with respect to }t,[\biota,\tau] \textrm{ and }x.
\end{align*}
So, $\mathfrak{S}_{\theta,\mathcal{A}}$ is a continuous-time linear skew-product semiflow driven by the continuous-time symbolic semiflow $\varTheta$.
\end{lem}

\begin{proof}
Based on the argument above, we need only to notice here that for all $t\in[0,\infty)$ there $\varTheta(t,[\biota,\tau])=[\biota,\tau+t]$ for any $[\biota,\tau]\in\mathrm{S}(\varSigma_{\!N}^+)$, and that for each $k\ge1$, $\biota_1=\biota_1^\prime,\dotsc,\biota_k=\biota_k^\prime$ as any two points $\biota, \biota^\prime$ are sufficiently close in $\varSigma_{\!N}^+$.
\end{proof}

This lemma is important for our discussion later, because combining with the above Lemma~\ref{lem3.1} it permits us to employ the classical ergodic theorems for
the stability of the linear switching systems (\ref{eq3.6}).

\section{Stabilizability of linear switching systems}\label{sec4}%

This section will be devoted to proving the statement (1) of Theorem~\ref{thm1.2} and Remark~\ref{rem1.3} stated in Section~\ref{sec1}. Let
\begin{equation*}
\mathcal{A}=\{A_1,\dotsc,A_N\}\subset\mathbb{R}^{n\times n}
\end{equation*}
be an arbitrary collection of real matrices. Then, it generates the linear switching dynamical systems
\begin{equation}\label{eq4.1}
\dot{x}(t)=A_{\sigma(t)}x(t),\quad x(0)\in\mathbb{R}^n\textrm{ and }t\in\mathbb{R}_+,
\end{equation}
where $\sigma\colon\mathbb{R}_+=(0,\infty)\rightarrow\{1,\dotsc,N\}$ are switching signals. For our goal, we need only to consider the special systems
\begin{equation*}
\dot{x}(t)=A_{\sigma_\biota(t)}x(t),\quad x(0)\in\mathbb{R}^n, t\in\mathbb{R}_+,\textrm{ and } \biota\in\varSigma_{\!N}^+ \leqno{(\ref{eq4.1})_\biota}
\end{equation*}
where $\sigma_\biota(t)$ is defined in the same way as in (\ref{eq1.8}); that is, $\sigma_\biota(t)\equiv\biota_k$ if $k-1<t\le k$ for all $k\in\mathbb{N}$.

We need an important classical result of triangularization for proving the statement (1) of Theorem~\ref{thm1.2}.

\begin{lem}[Lie's theorem~\cite{Hum}]\label{lem4.1}
If $\mathcal{A}$ is solvable over the field $\mathbbm{C}$ of complex numbers, then one can find a nonsingular complex matrix, say $\mathbb{T}\in\mathbbm{C}^{n\times n}$, such that
\begin{equation*}
\widetilde{A}_i=\mathbb{T}A_i\mathbb{T}^{-1}\quad\forall i\in\{1,\dotsc,N\}
\end{equation*}
all are upper-triangular.
\end{lem}

Note that even if the matrices $A_i$ have real entries, those of $\mathbb{T}$ defined by Lemma~\ref{lem4.1} may be complex~\cite{Hum}.

The statements (1) of Theorem~\ref{thm1.2} and Remark~\ref{rem1.3} then follow immediately from the following theorem.

\begin{thm}\label{thm4.2}
Let $\vec{\alpha}=(\alpha_1,\dotsc,\alpha_N)\in\mathbb{R}^N$ be a positive probability vector and assume that $\mathcal{A}=\{A_1,\dotsc,A_N\}\subset\mathbb{R}^{n\times n}$ is solvable over the field $\mathbbm{C}$. Then,
\begin{equation*}
\dot{x}(t)=(\alpha_1A_1+\cdots+\alpha_NA_N)x(t),\; x(0)\in\mathbb{R}^n\textrm{ and }t\in\mathbb{R}_+,\;\textrm{is exponentially stable} \leqno{(\star)}
\end{equation*}
if and only if for $\mathds{P}_{\vec{\alpha}}$-a.s.~$\biota\in\varSigma_{\!N}^+$, these
$(\ref{eq4.1})_{\biota}$ are exponentially stable.
\end{thm}

\begin{proof}
Let $\vec{\alpha}=(\alpha_1,\dotsc,\alpha_N)\in\mathbb{R}^N$ be arbitrarily given as in the assumption of the statement.
Let $\mathbb{T}$ be the matrix defined as in Lemma~\ref{lem4.1} by the solvability property of $\mathcal{A}$ so that
\begin{equation}\label{eq4.2}
\widetilde{A}_i=\mathbb{T}A_i\mathbb{T}^{-1}=\left[\begin{matrix}\tilde{a}_i^{11}&\cdots&*\\ \vdots&\ddots&\vdots\\ 0&\cdots&\tilde{a}_i^{nn}\end{matrix}\right]\quad \forall i\in\{1,\dotsc,N\}.
\end{equation}
And let
\begin{equation*}
\mathfrak{S}_{\theta,\mathcal{A}}\colon[0,\infty)\times\mathrm{S}(\varSigma_{\!N}^+)\times\mathbb{R}^n\rightarrow\mathrm{S}(\varSigma_{\!N}^+)\times\mathbb{R}^n
\end{equation*}
be the associated linear skew-product semiflow defined as in Lemma~\ref{lem3.2}. We first notice that for any $[\biota,\tau]\in\mathrm{S}(\varSigma_{\!N}^+)$, the stability of the system
\begin{equation}\label{eq4.3}
\dot{x}(t)=\mathds{A}_{\varTheta(t,[\biota,\tau])}x(t),\quad x(0)\in\mathbb{R}^n\textrm{ and }t\in\mathbb{R}_+
\end{equation}
is equivalent to that of the switching system
\begin{equation}\label{eq4.4}
\dot{x}(t)=A_{\sigma_\biota(\tau+t)}x(t),\quad x(0)\in\mathbb{R}^n\textrm{ and }t\in\mathbb{R}_+.
\end{equation}
In fact, (\ref{eq4.3}) and (\ref{eq4.4}) are the same equation from $(\ref{eq3.5})^\prime$.
We next consider the (maximal) Lyapunov exponent of (\ref{eq4.3}) given by
\begin{equation}\label{eq4.5}
\bbchi^+([\biota,\tau])=\limsup_{t\to+\infty}\frac{1}{t}\log\|\varPhi_{[\biota,\tau]}(t)\|.
\end{equation}
From
\begin{equation}\label{eq4.6}
\bbchi^+([\biota,\tau])=\limsup_{t\to+\infty}\frac{1}{t}\log\|\mathbb{T}\varPhi_{[\biota,\tau]}(t)\mathbb{T}^{-1}\|,
\end{equation}
it follows that for every $[\biota,\tau]\in\mathrm{S}(\varSigma_{\!N}^+)$, (\ref{eq4.3}) is exponentially stable if and only if so is the upper-triangular system
\begin{equation}\label{eq4.7}
\dot{x}(t)=\widetilde{\mathds{A}}_{\varTheta(t,[\biota,\tau])}x(t),\quad x(0)\in\mathbbm{C}^n\textrm{ and }t\in\mathbb{R}_+
\end{equation}
where
\begin{equation}\label{eq4.8}
\widetilde{\mathds{A}}_{[\biota,\tau]}=\widetilde{A}_{\sigma_\biota(\tau)}\quad\forall [\biota,\tau]\in\mathrm{S}(\varSigma_{\!N}^+).
\end{equation}
We should notice here that, under the variable transformation $z=\mathbb{T}x\in\mathbbm{C}^n$ for $x\in\mathbb{R}^n$,
\begin{equation*}
\widetilde{\varPhi}_{[\biota,\tau]}(t):=\mathbb{T}\varPhi_{[\biota,\tau]}(t)\mathbb{T}^{-1}\colon\mathbbm{C}^n\rightarrow\mathbbm{C}^n\quad\forall t\in\mathbb{R}_+
\end{equation*}
is the principal matrix of (\ref{eq4.7}), where $\varPhi_{[\biota,\tau]}(t)\colon\mathbb{R}^n\rightarrow\mathbb{R}^n$ is naturally complexified by \begin{equation*}
\varPhi_{[\biota,\tau]}(t)(x+\mathfrak{i}y)=\varPhi_{[\biota,\tau]}(t) x+\mathfrak{i}\varPhi_{[\biota,\tau]}(t) y\quad \forall x+\mathfrak{i}y\in \mathbbm{C}^n,
\end{equation*}
where $\mathfrak{i}=\sqrt{-1}$ is the imaginary unit; thus, $$\limsup_{t\to+\infty}\frac{1}{t}\log\|\widetilde{\varPhi}_{[\biota,\tau]}(t)\|=\limsup_{t\to+\infty}\frac{1}{t}\log\|\varPhi_{[\biota,\tau]}(t)\|.$$
Let
\begin{equation}\label{eq4.9}
\vartheta_i([\biota,\tau])=\limsup_{T\to+\infty}\frac{1}{T}\int_0^T\mathfrak{Re}(\tilde{a}_{\sigma_\biota(\tau+t)}^{ii})\,\mathrm{d}t\quad \textrm{for }i=1,\dotsc,n.
\end{equation}
Then from Theorem~\ref{thm2.1}, it follows that
\begin{equation}\label{eq4.10}
\bbchi^+([\biota,\tau])=\max\{\vartheta_i([\biota,\tau])\,|\,i=1,\dotsc,n\}.
\end{equation}
Define the qualitative functions
\begin{equation}\label{eq4.11}
\omega_i\colon\mathrm{S}(\varSigma_{\!N}^+)\rightarrow\mathbb{R};\quad[\biota,\tau]\mapsto\mathfrak{Re}(\tilde{a}_{\sigma_\biota(\tau)}^{ii})
\end{equation}
for each $i=1,\dotsc,n$. It is easy to see that each $\omega_i$ is bounded and Borel measurable such that
\begin{subequations}\label{eq4.12}
\begin{align}
&\omega_i\in\mathscr{F}_{\!\rho}(\mathrm{S}(\varSigma_{\!N}^+))&& \forall i=1,\dotsc,n \label{eq4.12a}\\
\intertext{and}
&\omega_i(\varTheta(t,[\biota,\tau]))=\mathfrak{Re}(\tilde{a}_{\sigma_\biota(\tau+t)}^{ii})&&\forall t\in[0,\infty)\textrm{ and } [\biota,\tau]\in\mathrm{S}(\varSigma_{\!N}^+),\label{eq4.12b}
\end{align}\end{subequations}
since $\varTheta(t,[\biota,\tau])=[\biota,\tau+t]$.
Thus, by (\ref{eq4.9}) and (\ref{eq4.12b}) we have
\begin{equation}\label{eq4.13}
\vartheta_i([\biota,\tau])=\limsup_{T\to+\infty}\frac{1}{T}\int_0^T\omega_i(\varTheta(t,[\biota,\tau]))\,\mathrm{d}t\quad \forall i=1,\dotsc,n
\end{equation}
for any $[\biota,\tau]\in\mathrm{S}(\varSigma_{\!N}^+)$.

Corresponding to the probability vector $\vec{\alpha}$, there are associated ergodic probability measures $\mathds{P}_{\vec{\alpha}}$ for the dynamics $(\varSigma_{\!N}^+,\theta)$ and $\mathcal{P}_{\vec{\alpha}}$ for $(\mathrm{S}(\varSigma_{\!N}^+),\varTheta)$ by Lemma~\ref{lem3.1}. Thus by the classical Birkhoff ergodic theorem~\cite{NS, Walters82}, it follows, from (\ref{eq4.13}) and (\ref{eq4.11}), that for $\mathcal{P}_{\vec{\alpha}}$-a.s.~$[\biota,\tau]\in\mathrm{S}(\varSigma_{\!N}^+)$ and for each $i=1,\dotsc,n$,
\begin{equation}\label{eq4.14}
\begin{split}
\vartheta_i([\biota,\tau])&=\int_{\mathrm{S}(\varSigma_{\!N}^+)}\omega_i([\biota,\tau])\,\mathrm{d}\mathcal{P}_{\vec{\alpha}}
=\int_{\mathrm{S}(\varSigma_{\!N}^+)}\mathfrak{Re}(\tilde{a}_{\sigma_\biota(\tau)}^{ii})\,\mathrm{d}\mathcal{P}_{\vec{\alpha}}([\biota,\tau])\\
&=\int_{\varSigma_{\!N}^+}\left\{\int_0^1\mathfrak{Re}(\tilde{a}_{\sigma_\biota(t)}^{ii})\,\mathrm{d}t\right\}\mathrm{d}\mathds{P}_{\vec{\alpha}}(\biota)\\
&=\int_{\varSigma_{\!N}^+}\mathfrak{Re}(\tilde{a}_{\biota_1}^{ii})\,\mathrm{d}\mathds{P}_{\vec{\alpha}}(\biota);
\end{split}
\end{equation}
here the last ``=" follows from $\sigma_\biota(t)=\biota_1$ for all $0<t\le1$ and any $\biota=(\biota_k)_{k=1}^{+\infty}\in\varSigma_{\!N}^+$. Therefore, from (\ref{eq1.7}) it follows that for $\mathcal{P}_{\vec{\alpha}}$-a.s. $[\biota,\tau]\in\mathrm{S}(\varSigma_{\!N}^+)$, we have
\begin{equation}\label{eq4.15}
\vartheta_i([\biota,\tau])=\sum_{j=1}^N\sum_{k=1}^N\mathfrak{Re}(\tilde{a}_{j}^{ii})\mathds{P}_{\vec{\alpha}}([k,j])=\sum_{j=1}^N\sum_{k=1}^N\alpha_k\alpha_j\mathfrak{Re}(\tilde{a}_{j}^{ii})=\sum_{j=1}^N\alpha_j\mathfrak{Re}(\tilde{a}_{j}^{ii})
\end{equation}
for every $i=1,\dotsc,n$, where $[k,j]\subset\varSigma_{\!N}^+$ is the cylinder set of length of $2$ defined by the word $(k,j)$ in the same manner as in (\ref{eq1.6}) in Section~\ref{sec1}.

\textit{Sufficiency.} Let condition $(\star)$ hold. Then the equation
$$\dot{z}=\mathbb{T}(\alpha_1A_1+\cdots+\alpha_NA_N)\mathbb{T}^{-1}z,\quad z\in\mathbbm{C}^n\textrm{ and }t\in\mathbb{R}_+$$
is exponentially stable. From
\begin{align*}
\mathbb{T}(\alpha_1A_1+\cdots+\alpha_NA_N)\mathbb{T}^{-1}&=\alpha_1\mathbb{T}A_1\mathbb{T}^{-1}+\cdots+\alpha_N\mathbb{T}A_N\mathbb{T}^{-1}\\
&=\left[\begin{matrix}\sum\limits_{j=1}^N\alpha_j\tilde{a}_{j}^{11}&\cdots&*\\
\vdots&\ddots&\vdots\\
0&\cdots&\sum\limits_{j=1}^N\alpha_j\tilde{a}_{j}^{nn}\end{matrix}\right],
\end{align*}
we have from Theorem~\ref{thm2.1}
\begin{equation}\label{eq4.16}
\sum_{j=1}^N\alpha_j\mathfrak{Re}(\tilde{a}_{j}^{ii})<0\quad \forall i=1,\dotsc,n.
\end{equation}
Thus from (\ref{eq4.10}), (\ref{eq4.15}) and (\ref{eq4.16}), it follows that
\begin{equation}\label{eq4.17}
\bbchi^+([\biota,\tau])<0\quad \textrm{for }\mathcal{P}_{\vec{\alpha}}\textrm{-a.s.}~[\biota,\tau]\in\mathrm{S}(\varSigma_{\!N}^+).
\end{equation}
Because for any $[\biota,\tau]\in\mathrm{S}(\varSigma_{\!N}^+)$ the exponential stability of (\ref{eq4.3}) is equivalent to that of $(\ref{eq4.1})_\biota$, from (\ref{eq4.17}) we see
that for $\mathcal{P}_{\vec{\alpha}}\textrm{-a.s.}~[\biota,\tau]\in\mathrm{S}(\varSigma_{\!N}^+)$, $(\ref{eq4.1})_\biota$ are exponentially stable.

Finally, let
$$\pi\colon\mathrm{S}(\varSigma_{\!N}^+)\rightarrow\varSigma_{\!N}^+;\quad [\biota,\tau]\mapsto\biota.$$
Then $\pi_*(\mathcal{P}_{\vec{\alpha}})=\mathds{P}_{\vec{\alpha}}$, i.e., $\mathds{P}_{\vec{\alpha}}=\mathcal{P}_{\vec{\alpha}}\circ\pi^{-1}$.
Thus, for $\mathds{P}_{\vec{\alpha}}\textrm{-a.s.}~\biota\in\varSigma_{\!N}^+$, $(\ref{eq4.1})_\biota$ are exponentially stable. This shows the sufficiency.

\textit{Necessity.} Let $(\ref{eq4.1})_{\biota}$ be exponentially stable for $\mathds{P}_{\vec{\alpha}}$-a.s.~$\biota\in\varSigma_{\!N}^+$. Then, we can obtain that for $\mathcal{P}_{\vec{\alpha}}$-a.s. $[\biota,\tau]\in\mathrm{S}(\varSigma_{\!N}^+)$, (\ref{eq4.3}) and hence (\ref{eq4.7}) are exponentially stable. So, it follows from Theorem~\ref{thm2.1} that for $\mathcal{P}_{\vec{\alpha}}$-a.s. $[\biota,\tau]\in\mathrm{S}(\varSigma_{\!N}^+)$, $\vartheta_i([\biota,\tau])<0$ for each $i=1,\dotsc,n$. Further, by (\ref{eq4.13}), Lemma~\ref{lem3.1} and the Birkhoff ergodic theorem, it follows from (\ref{eq4.15}) that
\begin{equation*}
\sum_{j=1}^N\alpha_j\mathfrak{Re}(\tilde{a}_{j}^{ii})<0\quad \textrm{for each }i=1,\dotsc,n.
\end{equation*}
Thus, from Theorem~\ref{thm2.1} once again, it follows that the equation
\begin{equation}\label{eq4.18}
\dot{z}(t)=\mathbb{T}(\alpha_1A_1+\cdots+\alpha_NA_N)\mathbb{T}^{-1}z(t),\quad z(0)\in\mathbbm{C}^n\textrm{ and }t\in\mathbb{R}_+
\end{equation}
is exponentially stable. So, the linear affine equation
\begin{equation*}
\dot{x}(t)=(\alpha_1A_1+\cdots+\alpha_NA_N)x(t),\quad x(0)\in\mathbb{R}^n\textrm{ and }t\in\mathbb{R}_+
\end{equation*}
is exponentially stable, by considering initial $z(0)\in \mathbb{T}(\mathbb{R}^n)$ for (\ref{eq4.18}). This completes the proof of the necessity.

This thus proves the theorem.
\end{proof}

From the proof above or directly from Theorem~\ref{thm4.2}, we can obtain the following, which is an important step toward the proof of the statement (2) of Theorem~\ref{thm1.2}.

\begin{prop}\label{prop4.3}
Assume that $\vec{\alpha}=(\alpha_1,\dotsc,\alpha_N)\in\mathbb{R}^N$ is a positive probability vector. If $\mathcal{A}=\{A_1,\dotsc,A_N\}\subset\mathbb{R}^{n\times n}$ is solvable over the complex field $\mathbbm{C}$ and the linear affine equation
\begin{equation*}
\dot{x}=(\alpha_1A_1+\cdots+\alpha_NA_N)x,\quad x\in\mathbb{R}^n \textrm{ and }t\in\mathbb{R}_+
\end{equation*}
is exponentially stable, then for $\mathcal{P}_{\vec{\alpha}}$-a.s.~$[\biota,\tau]\in\mathrm{S}(\varSigma_{\!N}^+)$,
\begin{equation}\label{eq4.19}
\dot{x}(t)=\mathds{A}_{\varTheta(t,[\biota,\tau])}x(t),\quad x(0)\in\mathbb{R}^n\textrm{ and }t\in\mathbb{R}_+
\end{equation}
are exponentially stable.
\end{prop}

\begin{remark}\label{rem4.4}
Up to here, we may ask if we can press forward without letup until we complete the proof of the statement (2) of Theorem~\ref{thm1.2}. Unfortunately, we cannot. Although we have proved until now that
for $\mathcal{P}_{\vec{\alpha}}$-a.s.~$[\biota,\tau]\in\mathrm{S}(\varSigma_{\!N}^+)$,
(\ref{eq4.19}) are exponentially stable, yet to proving the statement (2) of Theorem~\ref{thm1.2}, we essentially need to show that
for $\mathcal{P}_{\vec{\alpha}}$-a.s.~$[\biota,\tau]\in\mathrm{S}(\varSigma_{\!N}^+)$,
\begin{equation}\label{eq4.20}
\dot{x}=\mathds{A}_{\varTheta(t,[\biota,\tau])}x+f(x,t),\quad x(0)\in\mathbb{R}^n\textrm{ and }t\in\mathbb{R}_+
\end{equation}
are exponentially stable, where $\|f(x,t)\|\le \bL\|x\|$ for sufficiently small $\bL>0$. When $\|f(x,t)\|\le \bL\|x\|^{1+\gamma}$ for some $\gamma>0$, we could conclude our desirable result from the classical Lyapunov stability theorem~\cite{Lya}. However, the classical Perron counterexample~\cite{Perron} shows that the high order $\gamma$-condition is crucial. So, to overcome the trouble caused by the lack of such $\gamma$, we will select carefully out ``sufficiently many good" driving points $[\biota,\tau]$ among the $\mathcal{P}_{\vec{\alpha}}$-a.s. $[\biota,\tau]$, using another ergodic theorem presented in Section~\ref{sec5}.
\end{remark}

\section{Stabilizability of quasilinear switched systems}\label{sec5}
This section will be devoted to proving the statement (2) of Theorem~\ref{thm1.2} stated in Section~\ref{sec1}, based on the statement (1) of Theorem~\ref{thm1.2} proved in Section~\ref{sec4}.

Let $f_i(x,t)\in\mathbb{R}^n, i=1,\dotsc, N$, be measurable functions, which are continuous with respect to the state-variable $x\in\mathbb{R}^n$, such that the following linear growth condition holds:
\begin{equation}\label{eq5.1}
\|f_i(x,t)\|\le \bL\|x\|\quad \forall x\in\mathbb{R}^n
\end{equation}
uniformly for $t\in\mathbb{R}_+=(0,+\infty)$ and $i\in\{1,\dotsc,N\}$, for some constant $\bL>0$. To prove the statement (2) of Theorem~\ref{thm1.2}, it is sufficient to show the following slightly general stability result.

\begin{thm}\label{thm5.1}
Assume that $\vec{\alpha}=(\alpha_1,\dotsc,\alpha_N)\in\mathbb{R}^N$ is a positive probability vector. Let $\mathcal{A}=\{A_1,\dotsc,A_N\}\subset\mathbb{R}^{n\times n}$ be solvable over $\mathbbm{C}$, and let the linear affine equation
\begin{equation*}
\dot{x}=(\alpha_1A_1+\dotsm+\alpha_NA_N)x,\quad x\in\mathbb{R}^n\textrm{ and }t\in\mathbb{R}_+
\end{equation*}
be exponentially stable. Then for any sufficiently small $\varepsilon>0$, one can find a Borel subset $W\subset\varSigma_{\!N}^+$
with $\mathds{P}_{\vec{\alpha}}(W)\ge1-\varepsilon$ and a constant $\bbdelta>0$ such that for every $\biota\in W$, the switching systems
\begin{equation}
\dot{x}(t)=A_{\sigma_{\biota}(t)}x(t)+f_{\sigma_{\biota}(t)}(x(t), t),\quad x(0)\in\mathbb{R}^n\textrm{ and }t\in\mathbb{R}_+
\end{equation}
are globally exponentially stable, whenever $f_1(x,t), \dotsc, f_N(x,t)$ satisfy condition (\ref{eq5.1}) with $\bL<\bbdelta$.
\end{thm}

Indeed, for proving the statement (2) of Theorem~\ref{thm1.2}, we need only to use Theorem~\ref{thm5.1} with $f_i(x,t)=B_i(x)u(t)$ for reasonable external input $u(t)\in\mathbb{R}^{m_{\sigma_\biota(t)}}$.

Furthermore, by $(\ref{eq3.5})^\prime$ this theorem is equivalent to the following proposition.

{\sc \textbf{Proposition} 5.2}$^\prime.$
\;{\it Assume that $\vec{\alpha}=(\alpha_1,\dotsc,\alpha_N)\in\mathbb{R}^N$ is a positive probability vector. Let $\mathcal{A}=\{A_1,\dotsc,A_N\}\subset\mathbb{R}^{n\times n}$ be solvable over $\mathbbm{C}$ and the linear equation
\begin{equation*}
\dot{x}=(\alpha_1A_1+\cdots+\alpha_NA_N)x,\quad x(0)\in\mathbb{R}^n\textrm{ and }t\in\mathbb{R}_+\quad
\textrm{is exponentially stable.} \leqno{(\star)}
\end{equation*}
Then for any sufficiently small $\varepsilon>0$, one can find a Borel subset $Z\subset\mathrm{S}(\varSigma_{\!N}^+)$
with $\mathcal{P}_{\vec{\alpha}}(Z)\ge1-\varepsilon$ and a constant $\bbdelta>0$ such that for each driving point $[\biota,\tau]\in Z$, the switching system
\begin{equation*}
\dot{x}=A_{\sigma_{\biota}(\tau+t)}x+f_{\sigma_{\biota}(\tau+t)}(x, t),\quad x(0)\in\mathbb{R}^n\textrm{ and }t\in\mathbb{R}_+
\end{equation*}
is globally exponentially stable, if $f_1(x,t), \dotsc, f_N(x,t)$ satisfy (\ref{eq5.1}) with $\bL<\bbdelta$.
}%

Moreover, from Proposition~\ref{prop4.3} we easily see that this proposition follows from the following more general result.

\begin{prop}\label{prop5.2}
Assume $\vec{\alpha}=(\alpha_1,\dotsc,\alpha_N)\in\mathbb{R}^N$ is a positive probability vector, and let $\mathcal{A}=\{A_1,\dotsc,A_N\}\subset\mathbb{R}^{n\times n}$ be solvable over $\mathbbm{C}$. If the linear switched systems
\begin{equation*}
\dot{x}=A_{\sigma_\biota(\tau+t)}x,\quad x\in\mathbb{R}^n\textrm{ and } t\in\mathbb{R}_+
\end{equation*}
are exponentially stable for $\mathcal{P}_{\vec{\alpha}}$-a.s.~$[\biota,\tau]\in\mathrm{S}(\varSigma_{\!N}^+)$,
then for any sufficiently small $\varepsilon>0$, one can find a Borel subset $Z\subset\mathrm{S}(\varSigma_{\!N}^+)$
with $\mathcal{P}_{\vec{\alpha}}(Z)\ge1-\varepsilon$ and a constant $\bbdelta>0$ such that for each driving point $[\biota,\tau]\in Z$, the switching system
\begin{equation}\label{eq5.3}
\dot{x}=A_{\sigma_{\biota}(\tau+t)}x+f_{\sigma_{\biota}(\tau+t)}(x, t),\quad x\in\mathbb{R}^n\textrm{ and }t\in\mathbb{R}_+
\end{equation}
is globally exponentially stable, if $f_1(x,t), \dotsc, f_N(x,t)$ satisfy (\ref{eq5.1}) with $\bL<\bbdelta$.
\end{prop}

From here on, we let $\mathcal{A}=\{A_1,\dotsc,A_N\}$ be an arbitrarily given collection of real $n\times n$ matrices, $N\ge2$. The remaining part of this section will be devoted to proving Proposition~\ref{prop5.2}.

\subsection{Liao-type exponents and a criterion of stability}\label{sec5.1}%
In this subsection, by introducing the so-called Liao-type exponent,
we will provide a criterion of asymptotic
exponential stability for a kind of deterministic switching systems
that are defined by switching the following infinite number of non-autonomous
subsystems:
\begin{equation}\label{eq5.4}
\dot{x}=S_i(t)x+F_i(x,t),\quad(t,x)\in\mathbb{R}_+\times\mathbb{R}^n;\qquad
i\in\mathbb{N},
\end{equation}
where, for each index $i\in\mathbb{N}=\{1,2,\dotsc\}$,
$S_i(t)=\left[S_i^{jk}(t)\right]_{1\le j,k\le
n}\in\mathbb{R}^{n\times n}$ is continuous upper-triangular and $F_i(x,t)\in\mathbb{R}^n$
is continuous with respect to $x$ and Borel-measurable in $t$, such that
\begin{equation*}
\|S_i(t)x\|\le\bbalpha\|x\|\quad\textrm{and}\quad\|F_i(x,t)\|\le
\bbell(t)\|x\|\qquad\forall x\in\mathbb{R}^n
\end{equation*}
uniformly for $i\in\mathbb{N}$, where $\bbalpha, \bbell(t)$ both are
independent of the indices $i\in\mathbb{N}$. We note that is is not required that $F_i(x,t)$ be Lipschitz with respect to $x\in\mathbb{R}^n$.

Given any constant $\mathrm{T}_{\!*}>0$, we let $\bbsigma\colon\mathbb{R}_+\rightarrow\mathbb{N}$ be an arbitrarily
given $\mathrm{T}_{\!*}$-switching signal piecewise constant
with a switching-time sequence $\{T_k\}_0^\infty\colon
0=T_0<T_1<T_2<\cdots$ with $T_k\to+\infty$; that is to say,
\begin{equation*}
\bbsigma(t)\equiv i_k\in\mathbb{N} \textrm{ whenever }T_{k-1}<t\le T_k
\quad
\textrm{and}
\quad
T_k-T_{k-1}\le \mathrm{T}_{\!*}\quad\forall k\in\mathbb{N}.
\end{equation*}
Then, $\bbsigma(t)$ defines a quasi-linear switching system
\begin{equation*}
\dot{x}(t)=S_{\bbsigma(t)}(t)x+F_{\bbsigma(t)}(x, t),\qquad(t,x)\in\mathbb{R}_+\times\mathbb{R}^n.\leqno{\mathbf{(S,F)}_{\bbsigma}}
\end{equation*}
To prove Proposition~\ref{prop5.2} we will present a criterion for the exponential stability of this type of switched dynamical systems here.
\subsubsection{Liao-type exponents of $\mathbf{(S,F)}_{\bbsigma}$}\label{sec5.1.1}%
Let $\{\mathbf{k}_m\}_{m=0}^{+\infty}$ be an arbitrarily given integer sequence
such that
\begin{equation*}
\mathbf{k}_0=0\quad \textrm{and}\quad 1\le \mathbf{k}_{m}-\mathbf{k}_{m-1}\le\pmb{\Delta}
\quad\forall m\in\mathbb{N},
\end{equation*}
where $\pmb{\Delta}$ is a positive integer. According to \cite{DHX10, DHX11}, associated to this sequence $\{\mathbf{k}_m\}_0^\infty$,
the real number
\begin{equation*}
\bbchi_*^+(\mathbf{S}_{\bbsigma}):=\limsup_{m\to+\infty}\frac{1}{T_{\mathbf{k}_m}}\sum_{i=0}^{m-1}\max_{1\le
j\le
n}\left\{\int_{T_{\mathbf{k}_i}}^{T_{\mathbf{k}_{i+1}}}S_{\bbsigma(t)}^{jj}(t)\,\mathrm{d}t\right\}
\end{equation*}
is called a {\it Liao-type exponent} of $\mathbf{(S,F)}_{\bbsigma}$.

Clearly, to different integer sequences $\{\mathbf{k}_m\}_{m=0}^{+\infty}$, one may get
different Liao-type exponents for the same switched system $\mathbf{(S,F)}_{\bbsigma}$. See \cite{DHX10, DHX11}.

From the above definition, it is easy to see that
$\bbchi_*^+(\mathbf{S}_{\bbsigma})$ is independent of the
``perturbation term" $\mathbf{F}=\{F_i(x,t)\}_{i\in\mathbb{N}}$. Let us see the
linear switching system
\begin{equation*}
\dot{z}(t)=S_{\bbsigma(t)}(t)z(t),\qquad
z(0)\in\mathbb{R}^n\textrm{ and }t\in\mathbb{R}_+\leqno{\mathbf{S}_{\bbsigma}}
\end{equation*}
as the ``linear approximation" of $\mathbf{(S,F)}_{\bbsigma}$. Then from Theorem~\ref{thm2.1},
\begin{equation*}
\bbchi^+(\mathbf{S}_{\bbsigma})=\max_{1\le j\le
n}\left\{\limsup_{T\to+\infty}\frac{1}{T}\int_0^T
S_{\bbsigma(t)}^{jj}(t)\,\mathrm{d}t\right\}
\end{equation*}
is the (maximal) Lyapunov exponent of the linear system
$\mathbf{S}_{\bbsigma}$. From the definitions above, we have
$\bbchi^+(\mathbf{S}_{\bbsigma})\le\bbchi_*^+(\mathbf{S}_{\bbsigma})$. And, in general,
$\bbchi^+(\mathbf{S}_{\bbsigma})\lvertneqq\bbchi_*^+(\mathbf{S}_{\bbsigma})$. See \cite{DHX11} for an explicit example.
\subsubsection{Criterion of exponential stability}\label{sec5.1.2}%
We now can formulate the exponential stability criterion for $\mathbf{(S,F)}_{\bbsigma}$ via the so-called Liao-type exponent
as follows:

\begin{thm}[{\cite[Theorem~3.1]{DHX11}}]\label{thm5.3}
Let $S_i(t)\in\mathbb{R}^{n\times n}$ be upper-triangular for each
$i\in\mathbb{N}$, and assume that $\mathbf{S}_{\bbsigma}$ has the Liao-type
exponent $\bbchi_*^+(\mathbf{S}_{\bbsigma})<0$ associated to a $\pmb{\Delta}$-sequence
$\{\mathbf{k}_m\}_{m=0}^{+\infty}$. Then, there exists a constant $\bbdelta>0$ such that, whenever $\bbell(t)\le
\bbdelta$ for $t$ sufficiently large, the switching system
$\mathbf{(S,F)}_{\bbsigma}$ is globally, asymptotically, exponentially stable.
\end{thm}

\textit{Note: }Here the constant $\bbdelta>0$ is
independent of $\mathbf{F}=\{F_i(x,t)\}_{i\in\mathbb{N}}$ but it actually depends upon the constants $\bbalpha, \pmb{\Delta}, \mathrm{T}_{\!*}$, and the Liao-type exponent $\bbchi_*^+$.

In addition, this criterion overcomes the lack of the $\gamma$-condition mentioned in Remark~\ref{rem4.4} in Section~\ref{sec4}.

\begin{remark}\label{rem5.4}
We notice here that Lie's theorem (Lemma~\ref{lem4.1}) deduces a simple upper-triangularization of $\mathcal{A}$ via a complex coordinates transformation $\mathbb{T}$.
For complex $\mathbf{S}=\{S_i(t)\}_{i\in\mathbb{N}}$, the proof of Theorem~\ref{thm5.3} above presented in \cite{DHX11}, however, is invalid. So, we will need to pursue another natural and real upper-triangularization for $\mathcal{A}$ which having negative Liao-type exponents for sufficiently many switching signals $\sigma_\biota(t)$.
\end{remark}
\subsection{Frame skew-product semiflows and upper-triangularization}\label{se5.2}%
Let $\mathcal{A}=\{A_1,\dotsc,A_N\}\subset\mathbb{R}^{n\times n}$ be any given.
To use the criterion of stability in terms of Liao-type exponents, we will, in this subsection, introduce a natural triangularization for $\mathcal{A}$ different from Lie's theorem.

\subsubsection{Preliminaries}\label{sec5.2.1}%
For each matrix $A\in\mathcal{A}$, it gives rise to the time-invariant continuous-time dynamical system
\begin{equation*}
\dot{x}(t)=Ax(t),\quad x\in\mathbb{R}^n\textrm{ and }t\in\mathbb{R}. \leqno{(\mathbf{A})}
\end{equation*}
Equivalently, it induces a continuous-time linear flow on the state space $\mathbb{R}^n$:
\begin{equation*}
\varPhi_{\!A}\colon\mathbb{R}\times\mathbb{R}^n\rightarrow\mathbb{R}^n;\quad(t,x)\mapsto e^{tA}x. \leqno{(\mathbf{A})^\prime}
\end{equation*}
Sometimes, we identify $\varPhi_{\!A}(t,x)$ with $\varPhi_{\!A}(t)x$ for all $(t,x)\in\mathbb{R}\times\mathbb{R}^n$.
Importantly for our goal, it can further induce frame flows as follows.

Let $\vec{b}_1,\dotsc,\vec{b}_n\in\mathbb{R}^n$ be arbitrary $n$ vectors. Then $\bb=(\vec{b}_1,\dotsc,\vec{b}_n)$ is called an {\it orthogonal $n$-frame} of $\mathbb{R}^n$, provided that
$\bb$ forms an orthogonal basis of the vector space $\mathbb{R}^n$, i.e., $\langle\vec{b}_i,\vec{b}_j\rangle=0$ and $\|\vec{b}_i\|\not=0$ for $1\le i\not=j\le n$; if, in addition, $\|\vec{b}_i\|=1$ for every $1\le i\le n$, then $\bb$ is called an {\it orthonormal $n$-frame} of $\mathbb{R}^n$. Write $\digamma_{\!n}$ and $\digamma_{\!n}^\natural$ as the sets of all orthogonal and orthonormal $n$-frames $\bb$ of $\mathbb{R}^n$ inherited
topologies from $\mathbb{R}^{n\times n}$, respectively. Clearly, $\digamma_{\!n}^\natural$ is a subspace of $\digamma_{\!n}$ and moreover $\digamma_{\!n}^\natural$ is compact, but $\digamma_{\!n}$ is not.

For the convenience of our discussion later, we denote the classical Gram-Schmidt orthonormalization procedure by ``$\mathrm{Ort}^\natural$"; that is to say, for each collection of independent vectors $\bu=(\vec{u}_1,\dotsc,\vec{u}_n)$,
$\mathrm{Ort}^\natural(\bu)=(\vec{v}_1^\natural,\dotsc,\vec{v}_n^\natural)\in\digamma_{\!n}^\natural$ is defined in this manner:
\begin{align*}
\vec{v}_1^\natural&=\frac{\vec{u}_1}{\|\vec{u}_1\|},\\
\vec{v}_2^\natural&=\frac{\vec{u}_2-\langle\vec{u}_2,\vec{v}_1^\natural\rangle\vec{v}_1^\natural}{\|\vec{u}_2-\langle\vec{u}_2,\vec{v}_1^\natural\rangle\vec{v}_1^\natural\|},\\
\vdots&\quad\vdots\qquad\vdots\\
\vec{v}_n^\natural&=\frac{\vec{u}_n-\sum\limits_{k=1}^{n-1}\langle\vec{u}_n,\vec{v}_k^\natural\rangle\vec{v}_k^\natural}{\|\vec{u}_n-\sum\limits_{k=1}^{n-1}\langle\vec{u}_n,\vec{v}_k^\natural\rangle\vec{v}_k^\natural\|}.
\end{align*}
Similarly, the classical Gram-Schmidt orthogonalization procedure is denoted by ``$\mathrm{Ort}$"; that is to say, $\mathrm{Ort}(\bu)=(\vec{v}_1,\dotsc,\vec{v}_n)\in\digamma_{\!n}$ is defined by
\begin{align*}
\vec{v}_1&=\vec{u}_1,\\
\vec{v}_2&=\vec{u}_2-\langle\vec{u}_2,\vec{v}_1^\natural\rangle\vec{v}_1^\natural,\\
\vdots&\quad\vdots\qquad\vdots\\
\vec{v}_n&=\vec{u}_n-\sum_{k=1}^{n-1}\langle\vec{u}_n,\vec{v}_k^\natural\rangle\vec{v}_k^\natural.
\end{align*}
By induction on the dimension of the state space $\mathbb{R}^n$, we can represent the above Gram-Schmidt procedure in terms of matrices as follows: Under the canonical $n$-frame/basis of $\mathbb{R}^n$, $\be=(\vec{e}_1,\dotsc,\vec{e}_n)$, where
\begin{equation*}
\vec{e}_1=\left(\begin{matrix}1\\0\\ \vdots\\0\end{matrix}\right)\in\mathbb{R}^n,\;\dotsc,\;\vec{e}_n=\left(\begin{matrix}0\\ \vdots\\0\\ 1\end{matrix}\right)\in\mathbb{R}^n,
\end{equation*}
we may see $\bu=(\vec{u}_1,\dotsc,\vec{u}_n)$ as an $n\times n$ matrix with columns $\vec{u}_1,\dotsc,\vec{u}_n$ and also $\mathrm{Ort}^\natural(\bu)$ as an $n\times n$ matrix with columns
$\vec{v}_1^\natural,\dotsc, \vec{v}_n^\natural$; then
\begin{equation}\label{eq5.5}
\mathrm{Ort}^\natural(\bu)=\bu\varPsi^{-1},\quad \textrm{where }\varPsi\in\mathbb{R}^{n\times n} \textrm{ is unique and upper-triangular}.
\end{equation}
For any $n$-frame $\bb=(\vec{b}_1,\dotsc,\vec{b}_n)\in\digamma_{\!n}$ and any $t\in\mathbb{R}$, let $$\varPhi_{\!A}(t,\bb)=(\varPhi_{\!A}(t,\vec{b}_1),\dotsc,\varPhi_{\!A}(t,\vec{b}_n)),$$
which still is an independent collection of vectors, but not necessarily belong to $\digamma_{\!n}$. So, for any $\bb\in\digamma_{\!n}$ and any $t\in\mathbb{R}$, we let
\begin{equation*}
\mathfrak{F}_{\!A}(t,\bb)=\mathrm{Ort}(\varPhi_{\!A}(t,\bb))\in\digamma_{\!n}\quad \textrm{and} \quad\mathfrak{F}_{\!A}^\natural(t,\bb)=\mathrm{Ort}^\natural(\varPhi_{\!A}(t,\bb))\in\digamma_{\!n}^\natural.
\end{equation*}
One can easily check that
\begin{equation}\label{eq5.6}
\mathfrak{F}_{\!A}\colon\mathbb{R}\times\digamma_{\!n}\rightarrow\digamma_{\!n}\quad \textrm{and}\quad\mathfrak{F}_{\!A}^\natural\colon\mathbb{R}\times\digamma_{\!n}^\natural\rightarrow\digamma_{\!n}^\natural
\end{equation}
both are flows in the classical sense. Specially, we should observe that $\mathfrak{F}_{\!A}(t,\bb)$ and $\mathfrak{F}_{\!A}^\natural(t,\bb)$ both are smooth with respect to the time-variable $t\in\mathbb{R}$.

Inspirited by Liao's original work~\cite{Liao63} for $\mathrm{C}^1$-class vector fields on compact Riemannian manifolds, we next introduce the so-called qualitative functions associated to $\mathcal{A}$.

\begin{defn}\label{def5.5}
The continuous functions
\begin{align*}
&\Omega_k\colon\digamma_{\!n}^\natural\times\mathcal{A}\rightarrow\mathbb{R}\qquad(k=1,\dotsc,n)\\
\intertext{given by}
&\Omega_k(\bb,A)=\left.\frac{d}{dt}\right|_{t=0}\|\mathrm{col}_k(\mathfrak{F}_{\!A}(t,\bb))\|\quad\forall (\bb,A)\in\digamma_{\!n}^\natural\times\mathcal{A},
\end{align*}
are called the ``Liao qualitative functions of $\mathcal{A}$''.
\end{defn}

Here
\begin{equation*}
\mathrm{col}_k\colon\bb=(\vec{b}_1,\dotsc,\vec{b}_n)\mapsto\vec{b}_k\quad\forall k=1,\dotsc,n.
\end{equation*}
Note here that, since $\varPhi_{\!A}(t,x)$ is continuously differentiable with respect to $t\in\mathbb{R}$, the above continuous functions are well defined.  Since $\digamma_{\!n}^\natural$ is compact, all $\Omega_k$ are bounded. From the definition, we easily have
\begin{equation}\label{eq5.7}
\log\|\mathrm{col}_k(\mathfrak{F}_{\!A}(T,\bb))\|=\int_0^T\Omega_k(\mathfrak{F}_{\!A}^\natural(t,\bb),A)\, \mathrm{d}t\quad (k=1,\dotsc,n)
\end{equation}
for all $T\ge0$ and any $\bb\in\digamma_{\!n}^\natural$. This shows that these functions $\Omega_k$ are closely related to the Lyapunov exponents. In fact, $\limsup_{T\to\infty}\frac{1}{T}\int_0^T\Omega_1(\mathfrak{F}_{\!A}^\natural(t,\bb),A)\, \mathrm{d}t$ is just a Lyapunov exponent of $A$.

Under the canonical orthonormal basis
$\be=(\vec{e}_1,\dotsc,\vec{e}_n)$ of $\mathbb R^n$, we view $z$ in $\mathbb R^n$
as a column vector with components $z^1,\dotsc,z^n$ and
$\bb\in\digamma_{\!n}^\natural$ as an $n$-by-$n$ orthogonal matrix
with columns $\mathrm{col}_1\bb,\dotsc,\mathrm{col}_n\bb$,
successively; in addition, we sometimes identify a linear transformation of $\mathbb{R}^n$ into itself with an $n\times n$ matrix.

Given any orthonormal $n$-frame $\bb\in\digamma_{\!n}^\natural$, we define, by linear extension, the linear orthogonal
transformation of $\mathbb{R}^n$ into itself
\begin{equation}\label{eq5.8}
\mathbb T_{\bb}\colon \mathbb{R}^n\rightarrow
\mathbb{R}^n
\end{equation}
in the way: $\vec{e}_k\mapsto{\mathrm{col}}_k\bb$ for $1\le k\le
n$, such that
\begin{equation*}
\mathbb T_{\bb}(z)=\bb
z:=\sum_{k=1}^nz^k\mathrm{col}_k\bb\quad\textrm{and}\quad \|z\|=
\|\mathbb{T}_{\bb}(z)\|\quad\forall z\in\mathbb{R}^n.
\end{equation*}
Now we define a family of linear
transformations
\begin{equation}\label{eq5.9}
\varPsi_{\!{A,\bb}}(t,\bcdot)=\mathbb T_{\mathfrak{F}_{\!A}^\natural(t,\bb)}^{-1}\left(\varPhi_{\!A}(t,\mathbb{T}_{\bb}(\bcdot))\right)\colon\mathbb R^n\rightarrow\mathbb R^n\quad\forall
t\in\mathbb{R},
\end{equation}
where $\varPhi_{\!A}(t,\bcdot)$ is defined by the equation $(\mathbf{A})$ as before. Then there holds the following commutativity:
\begin{equation}\label{eq5.10}
\begin{CD}
\mathbb R^n@>{\varPsi_{\!{A,\bb}}(t,\bcdot)}>>\mathbb{R}^n\\
@V{\mathbb T_{\bb}}VV @VV {\mathbb{T}_{\mathfrak{F}_{\!A}^\natural(t,\bb)}}V\\
\mathbb R^n@>{\varPhi_{\!A}(t,\bcdot)}>>\mathbb R^n
\end{CD}\qquad\qquad \forall t\in\mathbb{R}.
\end{equation}
Equivalently, in terms of matrices,
\begin{equation*}
\varPhi_{\!A}(t,\bb)=\mathfrak{F}_{\!A}^\natural(t,\bb)\varPsi_{\!{A,\bb}}(t,\be)\quad\forall t\in\mathbb{R}.\leqno{(\ref{eq5.10})^\prime}
\end{equation*}
That is to say,
\begin{equation*}
\mathrm{Ort}^\natural(\varPhi_{\!A}(t,\bb))=\varPhi_{\!A}(t,\bb)\varPsi_{\!{A,\bb}}^{-1}(t,\be).
\end{equation*}
We now think of $\varPsi_{A,\bb}(t,\bcdot)$ as an $n$-by-$n$ nonsingular real
matrix. Then, $\varPsi_{\!{A,\bb}}^{-1}(t)$ is upper-triangular from (\ref{eq5.5}) and hence $\varPsi_{\!{A,\bb}}(t)$ is so. Clearly, $\frac{d}{dt}\varPsi_{A,\bb}(t,\bcdot)$ makes sense since
$\varPhi_{\!A}(t,\bcdot)$ is a smooth linear flow, and we have from (\ref{eq5.10})
\begin{equation}\label{eq5.11}
\varPsi_{\!{A,\bb}}(t+s,\bcdot)=\varPsi_{\!{A,{\mathfrak{F}_{\!A}^\natural(t,\bb)}}}\left(s,
\varPsi_{\!{A,\bb}}(t,\bcdot)\right)\quad\forall t,s\in\mathbb{R}.
\end{equation}
Put
\begin{subequations}\label{eq5.12}
\begin{align}
R_{A,\bb}(0)&=\left.\frac{d}{dt}\right|_{t=0}^{}\varPsi_{\!{A,\bb}}(t,\bcdot)\in\mathbb{R}^{n\times n}\\
\intertext{and} R_{A,\bb}(t)&=R_{A,\mathfrak{F}_{\!A}^\natural(t,\bb)}(0)\quad\forall t\in\mathbb{R}\\
\intertext{for any $\bb\in\digamma_{\!n}^\natural$. Then}
R_{A,\bb}(t+t_1)&=R_{A,\mathfrak{F}_{\!A}^\natural(t,\bb)}(t_1)\quad\forall t,t_1\in\mathbb{R}\label{eq5.12c}
\end{align}
\end{subequations}
for any $A\in\mathcal{A}$ and any $\bb\in\digamma_{\!n}^\natural$.

\begin{defn}\label{defn5.6}
For any $(A,\bb)\in \mathcal{A}\times\digamma_{\!n}^\natural$, the linear equation
\begin{equation*}
\dot{z}(t)=R_{A,\bb}(t)z(t),\quad(t,z)\in\mathbb{R}\times\mathbb{R}^n
\leqno{(R_{A,\bb})}
\end{equation*}
is called
the ``Liao triangularization system of $A$'' under the moving
frames $\mathfrak{F}_{\!A}^\natural(\mathbb{R},\bb)$.
\end{defn}

We will need the following basic results.

\begin{lem}\label{lem5.7}
Given any $(A,\bb)\in \mathcal{A}\times\digamma_{\!n}^\natural$,
the linear system $(R_{A,\bb})$ has the following three properties:
\begin{enumerate}
\item[$\mathrm{(1)}$] Upper-triangularity: for any $t\in\mathbb{R}$,
$R_{A,\bb}(t)$ is upper-triangular with diagonal elements
\begin{equation*}
R_{A,\bb}^{kk}(t)=\Omega_k(\mathfrak{F}_{\!A}^\natural(t,\bb),A)\quad \textrm{for }k=1,\dotsc,n.
\end{equation*}

\item[$\mathrm{(2)}$] Geometric property: for any $z_0\in\mathbb{R}^n$, $z(t,z_0)$ is the solution of
$(R_{A,\bb})$ satisfying the initial condition $z(0,z_0)=z_0$
if and only if
\begin{equation*}
\varPhi_{\!A}(t, \bb z_0)=\mathfrak{F}_{\!A}^\natural(t,\bb)z(t,z_0)
\quad\forall t\in\mathbb{R}.
\end{equation*}
Particularly, $\varPsi_{\!{A,\bb}}(t,\bcdot)$ is the principal matrix
of $(R_{A,\bb})$ with $\varPsi_{\!{A,\bb}}(0)=\mathrm{Id}_{\mathbb{R}^n}$.

\item[$\mathrm{(3)}$] Boundedness: there is a constant $\bbC<\infty$ such that
\begin{equation*}
\sup_{t\in\mathbb{R}}\left\{\sum_{1\le i,j\le n}|R_{A,\bb}^{ij}(t)|\right\}\le\bbC
\end{equation*}
uniformly for $\bb\in\digamma_{\!n}^\natural$.
\end{enumerate}
\end{lem}

\begin{proof}
This lemma comes immediately from \cite[Lemma~8]{Dai09} and so we omit the details here.
\end{proof}

As $\mathfrak{F}_{\!A}^\natural(t,\bb)$ is an orthonormal $n$-frame of $\mathbb{R}^n$, it follows from the statement (2) of Lemma~\ref{lem5.7} that $(\mathbf{A})$ has the same stability as $(R_{A,\bb})$ for any $\bb\in\digamma_{\!n}^\natural$.

\subsubsection{Real upper-triangularization of $\mathcal{A}$}\label{sec5.2.2}%
For any $A\in\mathcal{A}$, we have considered its induced dynamics in the last Subsection~\ref{sec5.2.1}. We will now turn to the switching of these dynamics. That is to say, we consider the linear dynamical system
\begin{equation*}
\mathfrak{S}_{\theta,\mathcal{A}}\colon[0,\infty)\times\mathrm{S}(\varSigma_{\!N}^+)\times\mathbb{R}^n\rightarrow\mathrm{S}(\varSigma_{\!N}^+)\times\mathbb{R}^n;\quad(t,([\biota,\tau],x))\mapsto([\biota,\tau+t], \varPhi_{[\biota,\tau]}(t,x)),
\end{equation*}
which is the linear skew-product semiflow, driven by the symbolic semiflow
\begin{equation*}
\varTheta\colon[0,\infty)\times\mathrm{S}(\varSigma_{\!N}^+)\rightarrow\mathrm{S}(\varSigma_{\!N}^+);\quad (t,[\biota,\tau])\mapsto[\biota,\tau+t],
\end{equation*}
defined in the same manner as in Lemma~\ref{lem3.2}. We notice here that $\varPhi_{[\biota,\tau]}(t)$ is piecewise smooth with respect to $t\in\mathbb{R}_+$.

Using the Gram-Schmidt procedures as in Subsection~\ref{sec5.2.1}, we have got that for all $[\biota,\tau]\in\mathrm{S}(\varSigma_{\!N}^+)$ with $0\le\tau<1$ and $\biota=(\biota_k)_{k=1}^{+\infty}\in\varSigma_{\!N}^+$,
\begin{equation}\label{eq5.13}
\mathfrak{F}_{\![\biota,\tau]}(t,\bb):=\mathrm{Ort}(\varPhi_{[\biota,\tau]}(t,\bb))\in\digamma_{\!n}\quad\forall \bb\in\digamma_{\!n}\textrm{ and }t>0
\end{equation}
is such that
\begin{equation*}
\mathfrak{F}_{\![\biota,\tau]}(t,\bb)=
\begin{cases}
\mathfrak{F}_{\!{A_{\biota_1}}}(t,\bb)& \textrm{if }0<t\le1-\tau,\\
\mathfrak{F}_{\!{A_{\biota_2}}}(t-1+\tau,\mathfrak{F}_{\!{A_{\biota_1}}}(1-\tau,\bb))& \textrm{if }1<t+\tau\le2,\\
\mathfrak{F}_{\!{A_{\biota_{k+1}}}}(t-k+\tau,\mathfrak{F}_{\!{A_{\biota_k}}}(1,\dotsc,\mathfrak{F}_{\!{A_{\biota_1}}}(1-\tau,\bb)\dotsc))& \textrm{if }k<t+\tau\le k+1;
\end{cases}
\end{equation*}
and
\begin{equation}\label{eq5.14}
\mathfrak{F}_{\![\biota,\tau]}^\natural(t,\bb):=\mathrm{Ort}^\natural(\varPhi_{[\biota,\tau]}(t,\bb))\in\digamma_{\!n}^\natural\quad\forall \bb\in\digamma_{\!n}^\natural
\end{equation}
is such that
\begin{equation*}
\mathfrak{F}_{\![\biota,\tau]}^\natural(t,\bb)=
\begin{cases}
\mathfrak{F}_{\!{A_{\biota_1}}}^\natural(t,\bb)& \textrm{if }t\in(0,1-\tau],\\
\mathfrak{F}_{\!{A_{\biota_2}}}^\natural(t-1+\tau,\mathfrak{F}_{\!{A_{\biota_1}}}^\natural(1-\tau,\bb))& \textrm{if }t\in(1,2]-\tau,\\
\mathfrak{F}_{\!{A_{\biota_{k+1}}}}^\natural(t-k+\tau,\mathfrak{F}_{\!{A_{\biota_k}}}^\natural(1,\dotsc,\mathfrak{F}_{\!{A_{\biota_1}}}^\natural(1-\tau,\bb)\dotsc))& \textrm{if }t\in(k,k+1]-\tau.
\end{cases}
\end{equation*}
As in (\ref{eq5.6}), one can easily observe that
\begin{equation}\label{eq5.15}
\mathfrak{F}\colon[0,\infty)\times\mathrm{S}(\varSigma_{\!N}^+)\times\digamma_{\!n}\rightarrow\mathrm{S}(\varSigma_{\!N}^+)\times\digamma_{\!n};\quad(t,([\biota,\tau],\bb))\mapsto([\biota,\tau+t],\mathfrak{F}_{\![\biota,\tau]}(t,\bb))
\end{equation}
{and}
\begin{equation}\label{eq5.16}
\mathfrak{F}^\natural\colon[0,\infty)\times\mathrm{S}(\varSigma_{\!N}^+)\times\digamma_{\!n}^\natural\rightarrow\mathrm{S}(\varSigma_{\!N}^+)\times\digamma_{\!n}^\natural; \quad(t,([\biota,\tau],\bb))\mapsto([\biota,\tau+t],\mathfrak{F}_{\![\biota,\tau]}^\natural(t,\bb))
\end{equation}
both are skew-product semiflows, called {\it frame skew-product semiflows}, still driven by the symbolic semiflow $\varTheta$ on $\mathrm{S}(\varSigma_{\!N}^+)$.

Given any $[\biota,\tau]\in\mathrm{S}(\varSigma_{\!N}^+)$ and $\bb\in\digamma_{\!n}^\natural$, similar to Subsection~\ref{sec5.2.1}, we can define a family of
linear isomorphisms/nonsingular matrices:
\begin{equation}\label{eq5.17}
\varPsi_{[\biota,\tau],\bb}(t,\bcdot)=\mathbb T_{\mathfrak{F}_{\![\biota,\tau]}^\natural(t,\bb)}^{-1}\varPhi_{[\biota,\tau]}(t)\mathbb{T}_{\bb}(\bcdot)\colon\mathbb R^n\rightarrow\mathbb R^n\quad\forall
t\in\mathbb{R}_+.
\end{equation}
Then there holds the following commutativity:
\begin{equation}\label{eq5.18}
\begin{CD}
\mathbb R^n@>{\varPsi_{[\biota,\tau],\bb}(t,\bcdot)}>>\mathbb{R}^n\\
@V{\mathbb T_{\bb}}VV @VV {\mathbb{T}_{\mathfrak{F}_{\![\biota,\tau]}^\natural(t,\bb)}}V\\
\mathbb R^n@>{\varPhi_{[\biota,\tau]}(t,\bcdot)}>>\mathbb R^n
\end{CD}\qquad\qquad\;\;\forall t\in\mathbb{R}_+.
\end{equation}
That is to say,
\begin{equation*}
\mathrm{Ort}^\natural(\varPhi_{[\biota,\tau]}(t,\bb))=\varPhi_{[\biota,\tau]}(t,\bb) \varPsi_{[\biota,\tau],\bb}^{-1}(t,\be)\quad\forall t\in\mathbb{R}_+.\leqno{(\ref{eq5.18})^\prime}
\end{equation*}
It is easily seen that $\varPsi_{[\biota,\tau],\bb}(t)$, as an $n$-by-$n$ nonsingular
matrix-valued function of $t$, is smooth at $t$ with $t+\tau\not\in\mathbb{Z}$. Thus, the left-hand side derivative $\frac{d^{-}}{dt}\varPsi_{[\biota,\tau],\bb}(t,\bcdot)$ makes sense at every $t>0$. In addition, the cocycle/semigroup property holds:
\begin{equation}\label{eq5.19}
\varPsi_{[\biota,\tau],\bb}(t+s)=\varPsi_{\mathfrak{F}^\natural(t,([\biota,\tau],\bb))}(s)
\varPsi_{[\biota,\tau],\bb}(t)\quad\forall t,s\in\mathbb{R}_+.
\end{equation}
Similar to (\ref{eq5.12}), we put
\begin{equation}\label{eq5.20}
\mathds{R}_{[\biota,\tau],\bb}(t)=\left\{\frac{d^-}{dt}\varPsi_{[\biota,\tau],\bb}(t)\right\}\varPsi_{[\biota,\tau],\bb}^{-1}(t)\quad\forall t>0,
\end{equation}
for any $[\biota,\tau]\in\mathrm{S}(\varSigma_{\!N}^+)$ with $0\le\tau<1$ and any $\bb\in\digamma_{\!n}^\natural$, noting that we identify $[\biota, \tau+k]$ with $[\theta^k(\biota),\tau]$ for all $k\in\mathbb{N}$.

\begin{defn}\label{defn5.8}
The linear equation
\begin{equation*}
\dot{z}(t)=\mathds{R}_{[\biota,\tau],\bb}(t)z(t),\quad z(0)\in\mathbb{R}^n\textrm{ and }t\in\mathbb{R}_+
\leqno{(\mathds{R}_{[\biota,\tau],\bb})}
\end{equation*}
for any $([\biota,\tau],\bb)\in\mathrm{S}(\varSigma_{\!N}^+)\times\digamma_{\!n}^\natural$, is called
the ``Liao triangularization system of $\mathcal{A}$'' under the moving
frames $\mathfrak{F}_{[\biota,\tau]}^\natural(\mathbb{R}_+,\bb)$.
\end{defn}

We will see that $(\mathds{R}_{[\biota,\tau],\bb})$ is just a switching system of the subsystems $(R_{A_k,\bb})$ for $k=1,\dotsc,N$.

For any $[\biota,\tau]\in\mathrm{S}(\varSigma_{\!N}^+)$, there is no loss of generality in assuming $0\le\tau<1$ from $(\ref{eq3.5})^\prime$. We then observe that
\begin{equation}\label{eq5.21}
\sigma_\biota(\tau+\bcdot)\colon\mathbb{R}_+\rightarrow\{1,\dotsc,N\}
\end{equation}
is such that
\begin{equation*}
\sigma_\biota(\tau+t)=\biota_k\quad \textrm{ if }(k-1)-\tau<t\le k-\tau\textrm{ and }k\in\mathbb{N},\leqno{(\ref{eq5.21})^\prime}
\end{equation*}
where $\biota=(\biota_k)_{k=1}^{+\infty}\in\varSigma_{\!N}^+$.

The following two lemmas are useful for the study of the linear system $(\mathds{R}_{[\biota,\tau],\bb})$.

\begin{lem}\label{lem5.9}
For any $([\biota,\tau],\bb)\in\mathrm{S}(\varSigma_{\!N}^+)\times\digamma_{\!n}^\natural$ with $0\le\tau<1$, there holds that
\begin{equation*}
\mathds{R}_{[\biota,\tau], \bb}(t)=\begin{cases}
R_{A_{\biota_1},\bb}(t)&\textrm{if }0<t\le1-\tau,\\
R_{A_{\biota_{k+1}},\mathfrak{F}_{\![\biota,\tau]}^\natural(k-\tau,\bb)}(t-(k-\tau))&\textrm{if }k-\tau<t\le k+1-\tau\textrm{ and }k\in\mathbb{N}
\end{cases}
\end{equation*}
where $R_{A,\bb}(t)$ is as in Definition~\ref{defn5.6}.
\end{lem}

\begin{proof}
Since $\varPhi_{[\biota,\tau]}(t)$ is the principal matrix of the switching system
\begin{equation*}
\dot{x}=A_{\sigma_\biota(\tau+t)}x,\quad x\in\mathbb{R}^n\textrm{ and }t\in\mathbb{R}_+,
\end{equation*}
we have from $(\ref{eq5.21})^\prime$
\begin{equation*}
\varPhi_{[\biota,\tau]}(t)=\varPhi_{A_{\biota_1}}(t)\quad \textrm{for }0<t\le1-\tau.
\end{equation*}
So, from (\ref{eq5.14}), (\ref{eq5.18}), (\ref{eq5.20}) and the statement (2) of Lemma~\ref{lem5.7}, we can obtain the first equality of the statement. Similarly, we can prove the second equality.
This thus proves Lemma~\ref{lem5.9}.
\end{proof}

\begin{lem}\label{lem5.10}
For any $([\biota,\tau],\bb)\in\mathrm{S}(\varSigma_{\!N}^+)\times\digamma_{\!n}^\natural$, there holds that
\begin{equation*}
\mathds{R}_{[\biota,\tau], \bb}(t_1+t_2)=\mathds{R}_{\mathfrak{F}^\natural(t_1,([\biota,\tau],\bb))}(t_2)\quad \forall t_1,t_2\in\mathbb{R}_+,
\end{equation*}
where $\mathfrak{F}^\natural(\bcdot,(\bcdot,\bcdot))$ is the frame skew-product semiflow as in (\ref{eq5.16}).
\end{lem}

\begin{proof}
Without loss of generality, let $0\le\tau<1$. We can find some $k\in\mathbb{N}$ such that $(k-1)-\tau<t_1+t_2\le k-\tau$.
In the case $k=1$, from Lemma~\ref{lem5.9} and (\ref{eq5.12c}), it follows that
\begin{equation*}
\mathds{R}_{[\biota,\tau],\bb}(t_1+t_2)=R_{A_{\biota_1},\bb}(t_1+t_2)
=R_{A_{\biota_1},\mathfrak{F}_{\!A_{\biota_1}}^\natural(t_1,\bb)}(t_2);
\end{equation*}
on the other hand, for $0<\tau+t_1+t\le 1$ we have
\begin{equation*}
\varPsi_{[\biota,\tau+t_1],\mathfrak{F}_{\!A_{\biota_1}}^\natural(t_1,\bb)}(t)=\varPsi_{A_{\biota_1},\mathfrak{F}_{\!A_{\biota_1}}^\natural(t_1,\bb)}(t)
\end{equation*}
and so by (\ref{eq5.20}) we have
\begin{equation*}
\mathds{R}_{\mathfrak{F}^\natural(t_1,([\biota,\tau],\bb))}(t_2)=\left\{\left.\frac{d^-}{dt}\right|_{t=t_2}\varPsi_{A_{\biota_1},\mathfrak{F}_{\!A_{\biota_1}}^\natural(t_1,\bb)}(t)\right\}\varPsi_{A_{\biota_1},\mathfrak{F}_{\!A_{\biota_1}}^\natural(t_1,\bb)}^{-1}(t_2);
\end{equation*}
then, Lemma~\ref{lem5.7}.(2) follows that $\mathds{R}_{[\biota,\tau],\bb}(t_1+t_2)=\mathds{R}_{\mathfrak{F}^\natural(t_1,([\biota,\tau],\bb))}(t_2)$. For the case $k\ge2$, the statement can be similarly proved by using the semigroup property.

This thus completes the proof of Lemma~\ref{lem5.10}.
\end{proof}

Recall that $\mathds{A}_{[\biota,\tau]}$ is the random matrix defined by (\ref{eq3.5}). From Lemmas~\ref{lem5.7}, \ref{lem5.10} and \ref{lem5.9}, we can obtain the following important results.

\begin{lem}\label{lem5.11}
Given any $([\biota,\tau],\bb)\in\mathrm{S}(\varSigma_{\!N}^+)\times\digamma_{\!n}^\natural$,
the linear system $(\mathds{R}_{[\biota,\tau],\bb})$ possesses the following properties:
\begin{enumerate}
\item[$\mathrm{(1)}$] Upper-triangularity: for any $t>0$, the matrix
$\mathds{R}_{[\biota,\tau],\bb}(t)$ is upper-triangular with diagonal elements
\begin{equation*}
\mathds{R}_{[\biota,\tau],\bb}^{kk}(t)=\Omega_k(\mathfrak{F}_{\![\biota,\tau]}^\natural(t,\bb),\mathds{A}_{[\biota,\tau+t]})\quad \textrm{for }k=1,\dotsc,n.
\end{equation*}

\item[$\mathrm{(2)}$] Geometric property: for any $z_0\in\mathbb{R}^n$, $z(t,z_0)$ is the solution of
$(\mathds{R}_{[\biota,\tau],\bb})$ satisfying the initial condition $z(0,z_0)=z_0$
if and only if
\begin{equation*}
\varPhi_{[\biota,\tau]}(t, \bb z_0
)=\mathfrak{F}_{\![\biota,\tau]}^\natural(t,\bb) z(t,z_0)
\quad\forall t\in\mathbb{R}_+.
\end{equation*}
Particularly, $\varPsi_{[\biota,\tau],\bb}(t,\bcdot)$ is the principal matrix
of $(\mathds{R}_{[\biota,\tau],\bb})$.

\item[$\mathrm{(3)}$] Boundedness: there is a constant $\bbC<\infty$ such that
\begin{equation*}
\sup_{t>0}\left\{\sum_{1\le i,j\le n}|\mathds{R}_{[\biota,\tau],\bb}^{ij}(t)|\right\}\le\bbC
\end{equation*}
uniformly for $\bb\in\digamma_{\!n}^\natural$ and $[\biota,\tau]\in\mathrm{S}(\varSigma_{\!N}^+)$.
\end{enumerate}
\end{lem}

\begin{proof}
The upper-triangularity follows from the statement (1) of Lemma~\ref{lem5.7}, Lemma~\ref{lem5.9} and Lemma~\ref{lem5.10}. The geometric property follows from (\ref{eq5.18}) and (\ref{eq5.20}). Finally, the boundedness comes from Lemmas~\ref{lem5.10}, \ref{lem5.9} and \ref{lem5.7}.(3).
This thus completes the proof of Lemma~\ref{lem5.11}.
\end{proof}

From the geometric property in the above statement, we can clearly see that for any driving point $[\biota,\tau]\in\mathrm{S}(\varSigma_{\!N}^+)$,
\begin{equation}\label{eq5.22}
\dot{x}(t)=A_{\sigma_\biota(\tau+t)}x(t),\quad x(0)\in\mathbb{R}^n\textrm{ and }t\in\mathbb{R}_+,
\end{equation}
is exponentially stable if and only if so is the Liao upper-triangular system
\begin{equation*}
\dot{z}(t)=\mathds{R}_{[\biota,\tau],\bb}(t)z(t),\quad z(0)\in\mathbb{R}^n\textrm{ and }t\in\mathbb{R}_+\leqno{(\mathds{R}_{[\biota,\tau], \bb})}
\end{equation*}
for any/some $n$-frame $\bb\in\digamma_{\!n}^\natural$.

Comparing with Lie's theorem, our upper-triangular $\mathds{R}_{[\biota,\tau],\bb}(t)$ is real; yet the cost that we pay is that $\mathcal{A}$ becomes time-dependent.
Importantly, from Lemma~\ref{lem5.10} we can define a natural skew-product system
\begin{equation}\label{eq5.23}
\dot{z}(t)=\mathds{R}_{\mathfrak{F}^\natural(t,([\biota,\tau],\bb))}(0)z(t),\quad z\in\mathbb{R}^n\textrm{ and } t\in\mathbb{R}_+,\quad([\biota,\tau],\bb)\in\mathrm{S}(\varSigma_{\!N}^+)\times\digamma_{\!n}^\natural,
\end{equation}
driven by the frame skew-product semiflow $\mathfrak{F}^\natural$ defined in (\ref{eq5.16}).

From Lemma~\ref{lem5.9}, we can obtain the following lemma, which explains $(\mathds{R}_{[\biota,\tau], \bb})$ in terms of switching time-dependent subsystems.

\begin{lem}\label{lem5.12}
Let $([\biota,\tau],\bb)\in\mathrm{S}(\varSigma_{\!N}^+)\times\digamma_{\!n}^\natural$ be any given where $0\le\tau<1$ and write $\bb(t)=\mathfrak{F}_{\![\biota,\tau]}^\natural(t,\bb)$ for all $t\in\mathbb{R}_+$. Put
\begin{equation*}
S_k(t)=\begin{cases}R_{A_{\biota_1},\bb}(t)& \textrm{for }k=1,\\
R_{A_{\biota_k},\bb(k-1-\tau)}(t)& \textrm{for }k=2,3,\dotsc;
\end{cases}
\end{equation*}
and
\begin{equation*}
\bbsigma_{[\biota,\tau]}\colon\mathbb{R}_+\rightarrow\mathbb{N};\quad \bbsigma_{[\biota,\tau]}(t)=k\; \textrm{if }(k-1)-\tau<t\le k-\tau\; \forall k\in\mathbb{N}.
\end{equation*}
Then, $(\mathds{R}_{[\biota,\tau], \bb})$ is the switching system
\begin{equation}\label{eq5.24}
\dot{z}(t)=S_{\bbsigma_{[\biota,\tau]}(t)}(t)z(t),\quad z(0)\in\mathbb{R}^n\textrm{ and }t\in\mathbb{R}_+
\end{equation}
and $\bbsigma_{[\biota,\tau]}$ is a $\mathrm{T}_{\!*}$-switching signal with $\mathrm{T}_{\!*}=1$.
\end{lem}

We will see that (\ref{eq5.3}) is equivalent to a perturbation of (\ref{eq5.24}) and for ``sufficiently many" driving points $[\biota,\tau]$, (\ref{eq5.24})
has negative Liao-type exponents. This will allow us apply Theorem~\ref{thm5.3} to proving Proposition~\ref{prop5.2}.

\subsection{A spectral theorem}\label{sec5.3}%
Let $\vec{\alpha}=(\alpha_1,\dotsc,\alpha_N)\in\mathbb{R}^N$ be a positive probability vector as in the statement of Proposition~\ref{prop5.2}.
Then, $(\mathrm{S}(\varSigma_{\!N}^+),\varTheta,\mathcal{P}_{\vec{\alpha}})$ is an ergodic semiflow from Lemma~\ref{lem3.1}.
Since $\mathrm{S}(\varSigma_{\!N}^+)\times\digamma_{\!n}^\natural$ is compact, from the lifting lemma of ergodic measures proved in \cite{Dai09}, we can easily obtain
the following ergodicity result.

\begin{lem}\label{lem5.13}
Let $\vec{\alpha}=(\alpha_1,\dotsc,\alpha_N)$ be a positive probability vector. Then there exists an ergodic probability measure $\mathfrak{P}_{\vec{\alpha}}$ on $\mathrm{S}(\varSigma_{\!N}^+)\times\digamma_{\!n}^\natural$ for the frame skew-product semiflow
\begin{equation*}
\mathfrak{F}^\natural\colon[0,\infty)\times\mathrm{S}(\varSigma_{\!N}^+)\times\digamma_{\!n}^\natural\rightarrow\mathrm{S}(\varSigma_{\!N}^+)\times\digamma_{\!n}^\natural
\end{equation*}
such that $\mathcal{P}_{\vec{\alpha}}$ is its marginal measure; that is to say, for any Borel set $B\subset\mathrm{S}(\varSigma_{\!N}^+)$, $\mathcal{P}_{\vec{\alpha}}(B)=\mathfrak{P}_{\vec{\alpha}}(B\times\digamma_{\!n}^\natural)$.
\end{lem}

\begin{proof}
Let
\begin{equation*}
\pi\colon\mathrm{S}(\varSigma_{\!N}^+)\times\digamma_{\!n}^\natural\rightarrow\mathrm{S}(\varSigma_{\!N}^+);\quad ([\biota,\tau],\bb)\mapsto[\biota,\tau],
\end{equation*}
be the natural bundle projection. It is continuous and surjective under the natural topologies. From definitions, there holds the following commutativity:
\begin{equation*}
\begin{CD}
\mathrm{S}(\varSigma_{\!N}^+)\times\digamma_{\!n}^\natural@>{\mathfrak{F}^\natural(t,\bcdot,\bcdot)}>>\mathrm{S}(\varSigma_{\!N}^+)\times\digamma_{\!n}^\natural\\
@V{\pi}VV@VV{\pi}V\\
\mathrm{S}(\varSigma_{\!N}^+)@>{\varTheta(t,\bcdot)}>>\mathrm{S}(\varSigma_{\!N}^+)
\end{CD}\qquad\qquad\forall t\in\mathbb{R}_+.
\end{equation*}
As $\mathrm{S}(\varSigma_{\!N}^+)\times\digamma_{\!n}^\natural$ is compact and $\mathcal{P}_{\vec{\alpha}}$ is an ergodic probability measure of $\varTheta$ on $\mathrm{S}(\varSigma_{\!N}^+)$ by Lemma~\ref{lem3.1}, it follows from \cite{Dai09} that there exists at least one ergodic probability measure, say $\mathfrak{P}_{\vec{\alpha}}$ for the semiflow $\mathfrak{F}^\natural$ on $\mathrm{S}(\varSigma_{\!N}^+)\times\digamma_{\!n}^\natural$, such that $\mathcal{P}_{\vec{\alpha}}=\mathfrak{P}_{\vec{\alpha}}\circ\pi^{-1}$, as desired.
\end{proof}

We notice here that \cite{Dai09} is for flows but same arguments still work for semiflows and discrete-time continuous transformations.

Based on Lemma~\ref{lem5.13}, we can obtain the following result, which is an important step towards proving Proposition~\ref{prop5.2}.

\begin{lem}\label{lem5.14}
Under the same context as Proposition~\ref{prop5.2}, the Liao upper-triangular systems $(\mathds{R}_{[\biota,\tau],\bb})$ are exponentially stable for $\mathfrak{P}_{\vec{\alpha}}$-a.s. $([\biota,\tau],\bb)\in\mathrm{S}(\varSigma_{\!N}^+)\times\digamma_{\!n}^\natural$.
\end{lem}

\begin{proof}
By the assumption of Proposition~\ref{prop5.2}, we see that for $\mathcal{P}_{\vec{\alpha}}$-a.s. $[\biota,\tau]\in\mathrm{S}(\varSigma_{\!N}^+)$,
\begin{equation*}
\dot{x}(t)=A_{\sigma_\biota(\tau+t)}x(t),\quad x(0)\in\mathbb{R}^n\textrm{ and }t\in\mathbb{R}_+
\end{equation*}
are exponentially stable. Then, the statement follows immediately from Lemma~\ref{lem5.13} and Lemma~\ref{lem5.11}.(2).
\end{proof}

Noting that $\{\varPsi_{[\biota,\tau],\bb}(t)\}_{t\in\mathbb{R}_+}$ is the principal matrix of $(\mathds{R}_{[\biota,\tau],\bb})$ and $\mathfrak{P}_{\vec{\alpha}}$ is ergodic, from Lemma~\ref{lem5.14} we easily obtain the following.

\begin{cor}\label{cor5.15}
Under the same context as Proposition~\ref{prop5.2}, there exists a constant $\bbchi_{\vec{\alpha}}^+<0$ such that
\begin{equation*}
\lim_{t\to+\infty}\frac{1}{t}\log\|\varPsi_{[\biota,\tau],\bb}(t)\|=\bbchi_{\vec{\alpha}}^+
\end{equation*}
for $\mathfrak{P}_{\vec{\alpha}}$-a.s. $([\biota,\tau],\bb)\in\mathrm{S}(\varSigma_{\!N}^+)\times\digamma_{\!n}^\natural$.
\end{cor}

\begin{proof}
Noting (\ref{eq5.23}), this is a simple direct result of the classical multiplicative ergodic theorem \cite{FK, Ose} and Lemma~\ref{lem5.14}.
\end{proof}

Recall that from Lemma~\ref{lem5.11}, the non-autonomous coefficient matrix of $(\mathds{R}_{[\biota,\tau],\bb})$ is $\mathds{R}_{[\biota,\tau],\bb}(t)$ that is real upper-triangular with diagonal elements
\begin{equation*}
\mathds{R}_{[\biota,\tau],\bb}^{kk}(t)=\Omega_k(\mathfrak{F}_{\![\biota,\tau]}^\natural(t,\bb),\mathds{A}_{[\biota,\tau+t]})\quad \textrm{for }k=1,\dotsc,n,
\end{equation*}
where $\mathds{A}_{[\biota,\tau+t]}=A_{\sigma_\biota(\tau+t)}$ as in (\ref{eq3.5}). For our convenience, we introduce the following concept.

\begin{defn}\label{def5.16}
For $\mathcal{A}=\{A_1,\dotsc,A_N\}\subset\mathbb{R}^{n\times n}$, the functions
\begin{equation}\label{eq5.25}
\bomega_k\colon\mathrm{S}(\varSigma_{\!N}^+)\times\digamma_{\!n}^\natural\rightarrow\mathbb{R};\quad([\biota,\tau],\bb)\mapsto\Omega_k(\bb,A_{\sigma_\biota(\tau)})
\end{equation}
for $k=1,\dotsc,n$, where $\Omega_k$ is as in Definition~\ref{def5.5}, are called the ``Liao qualitative functions'' of the skew-product system (\ref{eq5.23}).
\end{defn}

Then, from Lemma~\ref{lem5.11} there follows at once that the following holds.

\begin{lem}\label{lem5.17}
The functions $\bomega_k, k=1,\dotsc,n$, all are bounded Borel-measurable on $\mathrm{S}(\varSigma_{\!N}^+)\times\digamma_{\!n}^\natural$, such that
\begin{equation*}
\mathds{R}_{[\biota,\tau],\bb}^{kk}(t)=\bomega_k\left(\mathfrak{F}^\natural(t,([\biota,\tau],\bb))\right)
\end{equation*}
for all $t\in\mathbb{R}_+$ and any $([\biota,\tau],\bb)\in\mathrm{S}(\varSigma_{\!N}^+)\times\digamma_{\!n}^\natural$.
\end{lem}

\begin{proof}
We need only to prove the Borel measurability. Noting $\sigma_\biota(\tau)=\biota_1$ for all $0<\tau\le1$, this measurability follows from $(\ref{eq3.1})^\prime$.
Note here that the section $\varSigma_{\!N}^+\times\{0\}$ has null measure in $\mathrm{S}(\varSigma_{\!N}^+)$.
\end{proof}

Now, we can use the Birkhoff ergodic theorem again to obtain the following spectral theorem which presents an integral expression of the Lyapunov exponent $\bbchi_{\vec{\alpha}}^+$.

\begin{lem}\label{lem5.18}
Under the same context as Proposition~\ref{prop5.2}, it holds that
\begin{equation*}
\bbchi_{\vec{\alpha}}^+=\max\left\{\int_{\mathrm{S}(\varSigma_{\!N}^+)\times\digamma_{\!n}^\natural}\bomega_k\, \mathrm{d}\mathfrak{P}_{\vec{\alpha}}\,|\,k=1,\dotsc,n\right\}.
\end{equation*}
Here $\bbchi_{\vec{\alpha}}^+$ is given by Corollary~\ref{cor5.15}.
\end{lem}

\begin{proof}
From the Birkhoff ergodic theorem~\cite{NS,Walters82}, we see that for $k=1,\dotsc,n$,
\begin{equation*}
\int_{\mathrm{S}(\varSigma_{\!N}^+)\times\digamma_{\!n}^\natural}\bomega_k \, \mathrm{d}\mathfrak{P}_{\vec{\alpha}}=\lim_{T\to+\infty}\frac{1}{T}\int_0^T\bomega_k\left(\mathfrak{F}^\natural(t,([\biota,\tau],\bb))\right)\mathrm{d}t
\end{equation*}
for $\mathfrak{P}_{\vec{\alpha}}$-a.s. $([\biota,\tau],\bb)\in\mathrm{S}(\varSigma_{\!N}^+)\times\digamma_{\!n}^\natural$. Then, the statement follows immediately from Corollary~\ref{cor5.15}, Lemma~\ref{lem5.17} and Theorem~\ref{thm2.1}.
\end{proof}

Next, we will show that for $\mathfrak{P}_{\vec{\alpha}}$-a.s. $([\biota,\tau],\bb)\in\mathrm{S}(\varSigma_{\!N}^+)\times\digamma_{\!n}^\natural$, $(\mathds{R}_{[\biota,\tau],\bb})$ have got negative Liao-type exponents $\bbchi_*^+$. To this end, we need an other ergodic theorem.

\begin{thm}[\cite{DZ}]\label{thm5.19}
Let $\phi\colon[0,\infty)\times X\rightarrow X$ be a semiflow on a compact metrizable space $X$, which preserves a probability measure $\mu$, and assume
$\{t_i\}_{i=1}^\infty$ is an arbitrarily given real sequence with property
\begin{equation*}
t_1\ge1,\; t_{i+1}=2t_i\;\forall i\ge1.
\end{equation*}
Then, for any $n$ real-valued functions $f_k(\bcdot)\in \mathscr{L}_\mathbb{R}^1(X,\mu), k=1,\dotsc,n$, there exists a Borel subset of $\mu$-measure $1$, write as $\widehat{\Gamma}$, such that for all $x\in\widehat{\Gamma}$,
\begin{equation*}
f_k^*(x)=\lim_{T\to+\infty}\frac{1}{T}\int_0^Tf_k(\phi(t,x))\,\mathrm{d}t\quad (k=1,\dotsc,n)
\end{equation*}
and
\begin{equation*}
\lim_{i\to\infty}\left\{\lim_{\ell\to\infty}\frac{1}{\ell}\sum_{j=0}^{\ell-1}\max_{1\le k\le n}\left\{|f_k^*(x)-\frac{1}{t_i}\int_{jt_i}^{(j+1)t_i}f_k(\phi(t,x))\,\mathrm{d}t|\right\}\right\}=0.
\end{equation*}
Particularly, if $f_k(\bcdot)\in\mathscr{L}_\mathbb{R}^\infty(X,\mu)$ and $\mu$ is $\phi$-ergodic, then $\widehat{\Gamma}$ is $\phi$-invariant.
\end{thm}

This is a strengthened version of the classical Birkhoff ergodic theorem. We will apply it to the case where $X=\mathrm{S}(\varSigma_{\!N}^+)\times\digamma_{\!n}^\natural$,
$\phi=\mathfrak{F}^\natural$ and $f_k(\bcdot)=\bomega_k(\bcdot)$ for $k=1,\dotsc,n$. Let
\begin{equation}
\bbchi_k^+=\int_{\mathrm{S}(\varSigma_{\!N}^+)\times\digamma_{\!n}^\natural}\bomega_k\, \mathrm{d}\mathfrak{P}_{\vec{\alpha}}\quad\forall k=1,\dotsc,n.
\end{equation}
Then, $\bbchi_{\vec{\alpha}}^+=\max\{\bbchi_k^+\,|\,k=1,\dotsc, n\}$ from Lemma~\ref{lem5.18}. And as a result of Theorem~\ref{thm5.19}, we have the following corollary by choosing $t_i=2^{i-1}$ for all $i=1,2,\dotsc$ and letting $\mu=\mathfrak{P}_{\vec{\alpha}}$:

\begin{cor}\label{cor5.20}
Under the same context as Proposition~\ref{prop5.2}, there exists an $\mathfrak{F}^\natural$-invariant Borel subset $\varGamma$ of $\mathrm{S}(\varSigma_{\!N}^+)\times\digamma_{\!n}^\natural$ with $\mathfrak{P}_{\vec{\alpha}}$-measure $1$, such that
\begin{equation*}
\lim_{i\to\infty}\left\{\lim_{\ell\to\infty}\frac{1}{\ell}\sum_{j=0}^{\ell-1}\max_{1\le k\le n}\left\{|\bbchi_k^+-\frac{1}{2^{i-1}}\int_{j2^{i-1}}^{(j+1)2^{i-1}}\bomega_k(\mathfrak{F}^\natural(t,([\biota,\tau],\bb)))\,\mathrm{d}t|\right\}\right\}=0
\end{equation*}
for any $([\biota,\tau],\bb)\in\varGamma$.
\end{cor}

From this corollary, we now can choose the important Liao-type exponents as follows:

\begin{lem}\label{lem5.21}
Under the same context as Proposition~\ref{prop5.2}, for
any $\varepsilon>0$ sufficiently small, one can find an integer $\ell\ge1$ and a
Borel subset $Z_{\varepsilon}$ of $\mathrm{S}(\varSigma_{\!N}^+)\times\digamma_{\!n}^\natural$ such
that
\begin{enumerate}
\item[$\mathrm{(1)}$] $\mathfrak{P}_{\vec{\alpha}}(Z_{\varepsilon})\ge 1-\varepsilon$;

\item[$\mathrm{(2)}$] for any $([\biota,\tau],\bb)\in Z_{\varepsilon}$,
\begin{equation*}
\limsup_{m\to+\infty}\frac{1}{m2^\ell}\sum_{i=0}^{m-1}\max_{1\le j\le
n}\left\{\int_{i2^{\ell}}^{(i+1)2^{\ell}}\bomega_k(\mathfrak{F}^\natural(t,([\biota,\tau],\bb)))\,\mathrm{d}t\right\}
\le\bbchi_{\vec{\alpha}}^++\varepsilon.
\end{equation*}
\end{enumerate}
Note: Here $\ell$ can be sufficiently large and $\bbchi_{\vec{\alpha}}^+$ is defined as in Corollary~\ref{cor5.15}.
\end{lem}

\begin{proof}
This can be proved by an argument similar to that of \cite[Theorem~3.7]{Dai06}. So, we omit the details here.
\end{proof}

Because the qualitative functions $\bomega_k$ are bounded by Lemma~\ref{lem5.17}, we can improve the statement of Lemma~\ref{lem5.21} by choosing sufficiently large $\ell$, as follows:

\begin{cor}\label{cor5.22}
Under the same context as Proposition~\ref{prop5.2}, for
any $\varepsilon>0$ so small that $\bbchi_{\vec{\alpha}}^++\varepsilon<0$, one can find an integer $\ell\ge1$ and a
Borel subset $Z_{\varepsilon}$ of the driving space $\mathrm{S}(\varSigma_{\!N}^+)\times\digamma_{\!n}^\natural$ such
that
\begin{enumerate}
\item[$\mathrm{(1)}$] $\mathfrak{P}_{\vec{\alpha}}(Z_{\varepsilon})\ge 1-\varepsilon$;

\item[$\mathrm{(2)}$] for any $([\biota,\tau],\bb)\in Z_{\varepsilon}$ with $0\le\tau<1$, letting $\{T_k\}_{k=0}^{+\infty}$ be defined by
\begin{equation*}
T_0=0,\quad T_k=k-\tau\;\forall k=1,2,\dotsc,
\end{equation*}
we have
\begin{equation*}\begin{split}
\bbchi_*^+([\biota,\tau],\bb)&:=\limsup_{m\to+\infty}\frac{1}{T_{m2^\ell}}\sum_{i=0}^{m-1}\max_{1\le
j\le
n}\left\{\int_{T_{i2^{\ell}}}^{T_{(i+1)2^{\ell}}}\bomega_k(\mathfrak{F}^\natural(t,([\biota,\tau],\bb)))\,\mathrm{d}t\right\}\\
&<0.
\end{split}\end{equation*}
\end{enumerate}
\end{cor}

So, for any $([\biota,\tau],\bb)\in Z_{\varepsilon}$ with $0\le\tau<1$, $\bbchi_*^+([\biota,\tau],\bb)$ is just the Liao-type exponent of $(\mathds{R}_{[\biota,\tau],\bb})$ associated to the switching signal $\bbsigma$ defined as in Lemma~\ref{lem5.12} with the $\mathrm{T}_{\!*}$-switching-time sequence $\{T_k\}_{k=0}^{+\infty}$ and $\{\mathbf{k}_m\}_{m=0}^{+\infty}$, where $\mathbf{k}_m=m2^\ell$ for all $m\in\mathbb{Z}_+$ and $\mathrm{T}_{\!*}=1$.

\subsection{Proof of Proposition~\ref{prop5.2}}\label{sec5.4}%
Now we are ready to prove Proposition~\ref{prop5.2}, which implies Theorem~\ref{thm5.1} and further Theorem~\ref{thm1.2}.

\begin{proof}
For any $\varepsilon>0$ sufficiently small, let $\ell\ge1$ and $Z_{\varepsilon}\subset\mathrm{S}(\varSigma_{\!N}^+)\times\digamma_{\!n}^\natural$ be given by Corollary~\ref{cor5.22}.
Given any $([\biota,\tau],\bb)\in Z_{\varepsilon}$ with $0\le\tau<1$, we next consider the stability of the switching system
\begin{equation}\label{eq5.27}
\dot{x}(t)=A_{\sigma_{\biota}(\tau+t)}x(t)+f_{\sigma_{\biota}(\tau+t)}(x(t), t),\quad x(0)\in\mathbb{R}^n\textrm{ and }t\in\mathbb{R}_+,
\end{equation}
if $f_1(x,t), \dotsc, f_N(x,t)$ satisfy condition (\ref{eq5.1}) with $\bL$ sufficiently small. Here, as before, $\mathbb{R}_+=(0,+\infty)$.

Let $\bb(t)=\mathfrak{F}_{[\biota,\tau]}^\natural(t,\bb)$ as in (\ref{eq5.16}) and $G(t)=\mathbb{T}_{\bb(t)}^{-1}$ for all $t\in\mathbb{R}_+$, where $\mathbb{T}_{\bb(t)}$ is defined in the same way
as in (\ref{eq5.8}). Then, $G(t)$ is a family of orthogonal transformations and is piecewise smooth in $t$ with $G(0)=[\mathrm{col}_1\bb,\dotsc,\mathrm{col}_n\bb]^{-1}$. From the equation
\begin{equation}\label{eq5.28}
\dot{x}(t)=A_{\sigma_\biota(\tau+t)}x(t),\quad x(0)\in\mathbb{R}^n\textrm{ and }t\in\mathbb{R}_+
\end{equation}
via the nonautonomous coordinates transformations
\begin{equation}\label{eq5.29}
z(t)=G(t)x(t)\quad \forall t\in\mathbb{R}_+
\end{equation}
we can obtain the equation
\begin{equation}\label{eq5.30}
\dot{z}(t)=\left(\left\{\frac{d^-}{dt}G(t)\right\}G^{-1}(t)+G(t)A_{\sigma_\biota(\tau+t)}G^{-1}(t)\right)z(t),\quad z(0)\in\mathbb{R}^n\textrm{ and }t>0.
\end{equation}
So from Lemma~\ref{lem5.11}, it follows that
\begin{equation*}
\mathds{R}_{[\biota,\tau],\bb}(t)=\left\{\frac{d^-}{dt}G(t)\right\}G^{-1}(t)+G(t)A_{\sigma_\biota(\tau+t)}G^{-1}(t)\quad \forall t>0.
\end{equation*}
Write
\begin{equation}\label{eq5.31}
F_{\sigma_\biota(\tau+t)}(z,t)=G(t)f_{\sigma_\biota(\tau+t)}(G^{-1}(t)z,t).
\end{equation}
From (\ref{eq5.1}), it follows that
\begin{equation*}
\|F_{\sigma_\biota(\tau+t)}(z,t)\|\le \bL\|z\|.\leqno{(\ref{eq5.31})^\prime}
\end{equation*}
Then from (\ref{eq5.27}), under (\ref{eq5.29}) we have got the equation
\begin{equation}\label{eq5.32}
\dot{z}(t)=\mathds{R}_{[\biota,\tau],\bb}(t)z(t)+F_{\sigma_\biota(\tau+t)}(z(t),t),\quad z(0)\in\mathbb{R}^n\textrm{ and }t>0.
\end{equation}
Moreover, from $(\ref{eq5.21})^\prime$ and Lemma~\ref{lem5.12}, the equation (\ref{eq5.32}) becomes the following switching system
\begin{equation}\label{eq5.33}
\dot{z}(t)=S_{\bbsigma_{[\biota,\tau]}(t)}(t)z(t)+F_{\bbsigma_{[\biota,\tau]}(t)}(z(t),t),\quad z(0)\in\mathbb{R}^n\textrm{ and }t>0,
\end{equation}
which has the Liao-type exponent $\bbchi_*^+([\biota,\tau],\bb)$ from Corollary~\ref{cor5.22}.

Let $\bbC>0$ be given by Lemma~\ref{lem5.11}.(3) and $\bbchi_{\vec{\alpha}}^+<0$ by Corollary~\ref{cor5.15}. Applying Theorem~\ref{thm5.3} with $\bbalpha=\bbC, \pmb{\Delta}=2^\ell, \mathrm{T}_{\!*}=1$, and $\bbchi_*^+=\bbchi_*^+([\biota,\tau],\bb)$, there follows that one can find some constant
\begin{equation*}
\bbdelta=\bbdelta(\bbC, \bbchi_{\vec{\alpha}}^+, 2^\ell)>0
\end{equation*}
such that (\ref{eq5.33}) is globally exponentially stable if the constant $\bL\le\bbdelta$.

This completes the proof of Proposition~\ref{prop5.2}.
\end{proof}

\medskip
Then the statements of Theorem~\ref{thm1.2} hold.
\subsection*{Acknowledgment}
The author is very grateful to the anonymous referees
for their insightful comments on this
manuscript.



\begin{thebibliography}{99}
\bibitem{Abramov}
   \newblock {\sc L.\,M.~Abramov},
   \newblock {\it On the entropy of flows},
   \newblock Dokl. Akad. Nauk. SSSR, 128 (1959), pp.~873--876.

\bibitem{AL01}
   \newblock {\sc A.\,A.~Agrachev and D.~Liberzon},
   \newblock {\it Lie-algebraic stability criteria for switched systems},
   \newblock SIAM J. Control Optim., 40 (2001), pp.~253--269.

\bibitem{AI}
   \newblock {\sc Y.\,A.~Al'pin and K.\,D.~Ikramov},
   \newblock {\it Reducibility theorems for pairs of matrices as rational criteria},
   \newblock Linear Algebra Appl., 313 (2000), pp.~155--161.

\bibitem{Branicky98}
   \newblock {\sc M.\,S.~Branicky},
   \newblock {\it Multiple Lyapunov functions and other analysis tools for switched and hybrid systems},
   \newblock IEEE Trans. Automat. Control, 43 (1998), pp.~475--482.

\bibitem{BP}
   \newblock {\sc A.~Bressan and B.~Piccoli},
   \newblock {\it Introduction to the Mathematical Theory of Control},
   \newblock AIMS on Applied Math. Vol. 2, American Institute of Mathematical Science, 2007.

\bibitem{Dai06}
   \newblock {\sc X.~Dai},
   \newblock {\it Exponential stability of nonautonomous linear differential equations with linear perturbations by Liao methods},
   \newblock J. Differential Equations, {225} (2006), pp.~549--572.

\bibitem{Dai09}
   \newblock {\sc X.~Dai},
   \newblock {\it Integral expressions of Lyapunov exponents for autonomous ordinary differential systems},
   \newblock Sci. China Ser. A: Math., 52 (2009), pp.~195--216.

\bibitem{Dai10}
   \newblock {\sc X.~Dai},
   \newblock {\it Optimal state points of the subadditive ergodic theorem},
   \newblock Nonlinearity, 24 (2011), pp.~1565--1573.

\bibitem{DHX08}
   \newblock {\sc X.~Dai, Y.~Huang, and M.~Xiao},
   \newblock {\it Almost sure stability of discrete-time switched linear systems: A topological point of view},
   \newblock SIAM J. Control Optim., 47 (2008), pp.~2137--2156.

\bibitem{DHX10}
   \newblock {\sc X.~Dai, Y.~Huang, and M.~Xiao},
   \newblock {\it Criteria of stability for continuous-time switched systems by using Liao-type exponents},
   \newblock SIAM J. Control Optim., 48 (2009/10), pp.~3271--3296.

\bibitem{DHX11}
   \newblock {\sc X.~Dai, Y.~Huang, and M.~Xiao},
   \newblock {\it Stability of time-varying nonlinear switching systems under perturbations},
   \newblock Preprint, 2010, arXiv:1109.1102v2 [cs.SY] 20 Dec 2011.

\bibitem{DZ}
   \newblock {\sc X.~Dai and Z.-L.~Zhou},
   \newblock {\it A generalization of a theorem of Liao},
   \newblock Acta Math. Sin. (Engl. Ser.), 22 (2006), pp.~207--210.

\bibitem{DBPLA}
   \newblock {\sc R.\,A.~Decarlo, M.S.~Branicky, S.~Pettersson, B.~Lennartson, and P.J.~Antsaklis},
   \newblock {\it Perspectives and results on the stability and stabilizability of hybrid systems},
   \newblock in Proc. IEEE: Special Issue Hybrid Systems, 88 (2000), pp.~1069--1082.

\bibitem{FR}
   \newblock {\sc A.~Fryszowski and T.~Rze\v{z}uchowski},
   \newblock {\it Continuous version of Filippov-Wa\v{z}ewski relaxation theorem},
   \newblock J. Differential Equations, 94 (1991), pp.~254--265.

\bibitem{FK}
   \newblock {\sc H.~Furstenberg and H.~Kesten},
   \newblock {\it Products of random matrices},
   \newblock {Ann. Math. Statist., 31 (1960), pp.~457--469}.

\bibitem{Gurevich}
   \newblock {\sc B.\,M.~Gurevich},
   \newblock {\it Construction of increasing partitions for special flows},
   \newblock Theory Probab. Appl., 10 (1965), pp.~627--645.

\bibitem{Gurvits}
   \newblock {\sc L.~Gurvits},
   \newblock {\it Stability of discrete linear inclusion},
   \newblock Linear Algebra Appl., 231 (1995), pp.~47--85.

\bibitem{HBF}
   \newblock {\sc H.~Haimovich, J.\,H.~Braslavsky, and F.\,E.~Felicioni},
   \newblock {\it Feedback stabilisation of switching discrete-time systems via Lie-algebraic techniques},
   \newblock IEEE Trans. Automat. Control, 56 (2011), pp.~1129--1135.

\bibitem{Hum}
   \newblock {\sc J.\,E.~Humphreys},
   \newblock {\it Introduction to Lie Algebras and Representation Theory},
   \newblock GTM 9, Springer-Verlag, New York, 1972.

\bibitem{ISW}
   \newblock {\sc B.~Ingralls, E.\,D~Sontag, and Y.~Wang},
   \newblock {\it An infinite-time relaxation theorem for differential inclusions},
   \newblock Proc. Amer. Math. Soc., 131 (2003), pp.~487--499.

\bibitem{Laffey78}
   \newblock {\sc T.\,J.~Laffey},
   \newblock {\it Simultaneous triangularization of matrices\,---\,low rank case and the nonderogatory case},
   \newblock Linear Multilinear Algebras, 6 (1978), pp.~269--305.

\bibitem{Lalley}
   \newblock {\sc S.\,P.~Lalley},
   \newblock {\it Distribution of periodic orbits of symbolic and Axiom A flows},
   \newblock Adv. Appl. Math., 8 (1987), pp.~154--193.

\bibitem{Liao63}
   \newblock {\sc S.~Liao},
   \newblock {\it Certain ergodic properties of a differential system on a compact differentiable manifold},
   \newblock Acta Sci. Natur. Univ. Pekinensis, 9 (1963), pp.~309--327.

\bibitem{LHM99}
   \newblock {\sc D.~Liberzon, J.\,P.~Hespanha, and A.\,S.~Morse},
   \newblock {\it Stability of switched systems: a Lie-algebraic condition},
   \newblock Systems $\&$ Control Letters, 37 (1999), pp.~117--122.

\bibitem{LA09}
   \newblock {\sc H.~Lin and P.\,J.~Antsaklis},
   \newblock {\it Stability and stabilizability of switched linear systems: A survey of recent results},
   \newblock IEEE Trans. Automat. Control, 54 (2009), pp.~308--322.

\bibitem{Lya}
   \newblock {\sc A.~Lyapunov},
   \newblock {\it The General Problem of the Stability of Motion},
   \newblock Taylor $\&$ Francis, 1992.

\bibitem{Margaliot06}
   \newblock {\sc M.~Margaliot},
   \newblock {\it Stability analysis of switched systems using variational principles: An introduction},
   \newblock Automatica, 42 (2006), pp.~2059--2077.

\bibitem{ML06}
   \newblock {\sc M.~Margaliot and D.~Liberzon},
   \newblock {\it Lie-algebraic stability conditions for nonlinear switched systems and differential inclusions},
   \newblock Systems $\&$ Control Letters, 55 (2006), pp.~8--16.

\bibitem{NB94}
   \newblock {\sc K.\,S.~Narendra and J.~Balakrishnan},
   \newblock {\it A common Lyapunov function for stable LTI systems with commuting $A$-matrices},
   \newblock IEEE Trans. Automat. Control, 39 (1994), pp.~2469--2471.

\bibitem{NS}
   \newblock {\sc V.\,V.~Nemytskii and V.\,V.~Stepanov},
   \newblock {\it Qualitative Theory of Differential Equations},
   \newblock Princeton University Press, Princeton, New Jersey 1960.

\bibitem{Ose}
   \newblock {\sc V.\,I.~Oseledec},
   \newblock {\it A multiplicative ergodic theorem, Lyapunov characteristic numbers for dynamical systems},
   \newblock Trudy Mosk Mat. Obsec., 19 (1968), pp.~119--210.

\bibitem{Perron}
   \newblock {\sc O.~Perron},
   \newblock {\it Die Ordnunfszahlen linearer Differentialglwichungssyteme},
   \newblock Math. Zs., 31 (1930), pp.~748--766.

\bibitem{RR}
   \newblock {\sc H.~Radjavi and P.~Rosenthal},
   \newblock {\it Simultaneous Triangularization},
   \newblock Springer-Verlag, New York 2000.

\bibitem{SNS}
   \newblock {\sc H.~Shim, D.\,J.~Noh, and J.\,H.~Seo},
   \newblock {\it Common Lyapunov function for exponentially stable nonlinear systems},
   \newblock J. Korean Institute of Electrical Engineers, 11 (2001), pp.~108--111.

\bibitem{Sontag}
   \newblock {\sc E.\,D.~Sontag},
   \newblock {\it Mathematical Control Theory: Deterministic Finite-dimensional Systems},
   \newblock 2nd edition, TAM {6}. Springer-Verlag, New York Tokyo 1998.

\bibitem{Sun04}
   \newblock {\sc Z.~Sun},
   \newblock {\it Stabilizability and insensitivity of switched linear systems},
   \newblock IEEE Trans. Automat. Control, 49 (2004), pp.~1133--1137.

\bibitem{Sun06}
   \newblock {\sc Z.~Sun},
   \newblock {\it Stabilization and optimization of switched linear systems},
   \newblock Automatica, 42 (2006), pp.~783--788.

\bibitem{Tok87}
   \newblock {\sc J.~Tokarzewski},
   \newblock {\it Stability of periodically switched linear systems and the switching frequency},
   \newblock Int. J. Systems Sci., 18 (1987), pp.~697--726.

\bibitem{Utkin77}
   \newblock {\sc V.\,I.~Utkin},
   \newblock {\it Variable structure systems with sliding modes},
   \newblock IEEE Trans. Automat. Control, 22 (1977), pp.~212--222.

\bibitem{Walters82}
   \newblock {\sc P.~Walters},
   \newblock {\it An Introduction to Ergodic Theory}, GTM 79,
   \newblock Springer-Verlag, New York, 1982.

\bibitem{WPD}
   \newblock {\sc M.\,A.~Wicks, P.~Peleties, and R.\,A.~DeCarlo},
   \newblock {\it Switched controller synthesis for the quadratic stabilization of a pair of unstable linear systems},
   \newblock Eur. J. Control, 4 (1998), pp.~140--147.
\end{thebibliography}
\end{document}